\documentclass[12pt]{article}

\usepackage{amsthm}
\usepackage{amssymb}
\usepackage{amsmath}
\usepackage{graphicx}
\usepackage{varioref}
\usepackage{color}

\newcommand{\R}{\mathbb{R}}
\newcommand{\N}{\mathbb{N}}

\newtheorem{theorem}[equation]{Theorem}
\newtheorem{corollary}[equation]{Corollary}
\newtheorem{lemma}[equation]{Lemma}
\newtheorem{proposition}[equation]{Proposition}
\newtheorem{remark}[equation]{Remark}
\newtheorem{conjecture}[equation]{Conjecture}

\def\be{\begin{equation}}
\def\ee{\end{equation}}
\def\ba{\begin{eqnarray}}
\def\ea{\end{eqnarray}}

\title{Predictability of subluminal and superluminal wave equations}

\author{Felicity C. Eperon\footnote{fce21@cam.ac.uk} and Harvey S. Reall\footnote{hsr1000@cam.ac.uk} \\ {\footnotesize Department of Applied Mathematics and Theoretical Physics, University of Cambridge,} \\ {\footnotesize Wilberforce Road, Cambridge CB3 0WA, UK} \\ \\ Jan J. Sbierski\footnote{Jan.Sbierski@maths.ox.ac.uk} \\ {\footnotesize Mathematical Institute, University of Oxford, Woodstock Road, Oxford OX2 6GG, UK} }

\begin{document}

\maketitle

\begin{abstract}
It is sometimes claimed that Lorentz invariant wave equations which allow superluminal propagation exhibit worse predictability than subluminal equations. To investigate this, we study the Born-Infeld scalar in two spacetime dimensions. This equation can be formulated in either a subluminal or a superluminal form. Surprisingly, we find that the subluminal theory is less predictive than the superluminal theory in the following sense. For the subluminal theory, there can exist multiple maximal globally hyperbolic developments arising from the same initial data. This problem does not arise in the superluminal theory, for which there is a unique maximal globally hyperbolic development. For a general quasilinear wave equation, we prove theorems establishing why this lack of uniqueness occurs, and identify conditions on the equation that ensure uniqueness. In particular, we prove that superluminal equations always admit a unique maximal globally hyperbolic development. In this sense, superluminal equations exhibit better predictability than generic subluminal equations. 
\end{abstract}

\tableofcontents

\section{Introduction}

Many Lorentz invariant classical field theories permit superluminal propagation of signals around non-trivial background solutions. It is sometimes claimed that such theories are unviable because the superluminality can be exploited to construct causality violating solutions, i.e., ``time machines". The argument for this is to consider two lumps of non-trivial field with a large relative boost: it is claimed that there exist solutions of this type for which small perturbations will experience closed causal curves \cite{adams}. However, this argument is heuristic: the causality-violating solution is not constructed, it is simply asserted to exist. This means that the argument is open to criticism on various grounds \cite{mukhanov,geroch,Papallo:2015rna}. 

The reason that causality violation would be problematic is that it implies a breakdown of predictability. In this paper, rather than focusing on causality violation, we will investigate predictability. Our aim is to determine whether there is any qualitative difference in predictability between Lorentz invariant classical theories which permit superluminal propagation and those that do not. 

We will consider quasilinear scalar wave equations for which causality is determined by a metric $g(u,du)$ which depends on the scalar field $u$ and its first derivative $du$. In the initial value problem we specify initial data $(S,u,du)$ where $S$ is the initial hypersurface and $u$, $du$ are chosen on $S$ such that $S$ is spacelike w.r.t. $g(u,du)$. We can now ask: what is the largest region $M$ of spacetime in which the solution is uniquely determined by the data on $S$? Uniqueness requires that $(M,g)$ should be globally hyperbolic with Cauchy surface $S$, i.e., the solution should be a globally hyperbolic development (GHD) of the data on $S$. This suggests that the ``largest region in which the solution is unique" will be a GHD that is inextendible as a GHD, i.e., it is a {\it maximal} globally hyperbolic development (MGHD). 

Our aim, then, is determine whether there is any qualitative difference between MGHDs for subluminal and superluminal equations. 

In section \ref{general}, we will introduce the class of scalar wave equations that we will study, and define what we mean by ``subluminal" and ``superluminal" equations. Note that the standard linear wave equation is both subluminal and superluminal according to our definition. 

In section \ref{BI} we will study an example of a Lorentz invariant equation in $1+1$ dimensions, namely the Born-Infeld scalar field. The general solution of this equation is known \cite{Barbashov:1966frq,Barbashov:1967avx}. This equation can be formulated in either a subluminal or superluminal form. One can consider the interaction of a pair of wavepackets in these theories. If the amplitude of the wavepackets is not too large then the wavepackets merge, interact, and then separate again \cite{Barbashov:1967avx}. In the subluminal theory they emerge with a time delay, in the superluminal theory there is a time advance. The MGHD is the entire 2d Minkowski spacetime in both cases. 

For larger amplitude, it is known that the solution can form singularities in the subluminal theory \cite{Barbashov:1967avx}. Singularities can also form in the superluminal theory. In both cases, the formation of a singularity leads to a loss of predictability because MGHDs are extendible across a Cauchy horizon, and the solution is not determined uniquely beyond a Cauchy horizon. However, there is a qualitative difference between the subluminal and superluminal theories. In the superluminal theory there is a {\it unique} MGHD. However, in the subluminal theory, MGHDs are not unique: {\it there can exist multiple distinct MGHDs arising from the same initial data}. 

This is worrying behaviour. Given a solution defined in some region $U$, we can ask: in which subset of $U$ is the solution determined uniquely by the initial data? In the superluminal case, this region is simply the intersection of $U$ with the unique MGHD, or, equivalently, the domain of dependence of the initial surface within $U$. This can be determined {\it from the solution itself}. However, in the subluminal case there is, in general, no such method of determining the appropriate subset of $U$. To determine the region in which the solution is unique, one has to construct all other solutions arising from the same initial data!

In section \ref{theorems} we will discuss the existence and uniqueness of MGHDs for a large class of quasilinear wave equations (in any number of dimensions). We start by proving a theorem asserting that two GHDs defined in regions $U_1$ and $U_2$ will agree in $U_1 \cap U_2$ {\it provided $U_1 \cap U_2$ is connected}. Thus if one can show that $U_1 \cap U_2$ is {\it always} connected then one always has uniqueness. We will prove that this is the case for any equation with the property that there exists a vector field which is timelike w.r.t. $g(u,du)$ for all $(u,du)$. For such an equation, and for a suitable initial surface, we prove that there exists a unique MGHD. Note that {\it any} superluminal equation admits such a vector field so {\it for any superluminal equation there exists a unique MGHD}. 

Our Born-Infeld example demonstrates that one cannot expect a unique MGHD for a general subluminal equation. One can define the maximal region in which solutions are unique, which we call the maximal unique globally hyperbolic development (MUGHD). Unfortunately, as mentioned above, there is no simple characterization of the MUGHD: given a solution defined in a region $U$, there is no simple general method for determining which part of $U$ belongs to the MUGHD. As we will show, one can establish some partial results e.g. for a solution defined in $U$, the solution is unique in the subset of $U$ corresponding to the domain of dependence of the initial surface detemined w.r.t. the {\it Minkowski} metric. However, this is rather a weak result especially for equations with a speed of propagation considerably less than the speed of light. 

An important application of the notion of a MGHD is Christodoulou's work on shock formation in relativistic perfect fluids \cite{Chr}. Given that this work concerns subluminal equations, one might wonder whether the MGHD constructed in Ref. \cite{Chr} suffers from the lack of uniqueness dicussed above. We will prove that if a MGHD ``lies on one side of its boundary" then it is unique. This provides a method for demonstrating uniqueness of a MGHD once it has been constructed. In particular, this implies that there is a unique MGHD for the initial data considered in Ref. \cite{Chr}. However, we emphasize that the equations of Ref. \cite{Chr} are likely to exhibit non-uniqueness of MGHDs for more complicated choices of initial data. 

Of course we have not answered the question which motivated the present work, namely whether it is possible to ``build a time machine" in any Lorentz invariant theory which admits superluminal propagation. However, our work does show that the object that one would have to study in order to address this question, namely the MGHD, is well-defined in a superluminal theory. Smooth formation of a time machine would require that there exist generic initial data belonging to some suitable class (e.g. smooth, compactly supported, data specified on a complete surface extending to spatial infinity in Minkowski spacetime) for which the MGHD is extendible, with a compactly generated \cite{hawking} Cauchy horizon.\footnote{The word ``generic" is included to reflect the condition that the time machine should be stable under small perturbations of the initial data.} In the Appendix we explain why this is not possible in $1+1$ dimensions. Whether this is possible in a higher dimensional superluminal theory (let alone {\it all} such theories) is an open question. 

\section{General scalar equation}

\label{general}

\subsection{Subluminal and superluminal equations}

\label{sec:subsuper}

Consider a scalar field $u: \mathbb{R}^{d+1} \rightarrow \mathbb{R}$ in $(d+1)$-dimensional Minkowski spacetime. Assume that the field satisfies a quasilinear equation of motion\footnote{Everything we say in the next few sections applies also to a quasilinear {\it system}, where $u$ denotes a $N$-component vector of scalar fields.}
\be
\label{eqndims}
 g^{\mu\nu}(u,du) \partial_\mu \partial_\nu u = F(u,du)
\ee
where $F$ is a smooth\footnote{Here, and throughout this paper, `smooth' means $C^\infty$.} function and \eqref{eqndims} is written with respect to the canonical coordinates $x^\mu$ on $\mathbb{R}^{d+1}$.

We will say that $(M,u)$ is a {\it hyperbolic solution} if $M$ is a connected open subset of $\mathbb{R}^{d+1}$ and $u:M \rightarrow \mathbb{R}$ is a smooth solution of the above equation for which $g^{\mu\nu}(u,du)$ has Lorentzian signature. For such a solution we can define $g_{\mu\nu}(u,du)$ as the inverse of $g^{\mu\nu}$ and then $(M,g)$ is a spacetime. Causality for the scalar field is determined by the metric $g$ so we will be studying the causal properties of the spacetime $(M,g)$. 

Now assume that we have a Minkowski metric $m_{\mu\nu}$ on $\mathbb{R}^{d+1}$ (i.e. a flat, Lorentzian metric), with inverse $m^{\mu\nu}$. We call the above equation {\it subluminal} if, whenever $g^{\mu\nu}$ is Lorentzian, every vector that is causal w.r.t. $g_{\mu\nu}$ is also causal w.r.t. $m_{\mu\nu}$ (so the null cone of $g_{\mu\nu}$ lies on, or inside, the null cone of $m_{\mu\nu}$). We call the equation {\it superluminal} if, whenever $g^{\mu\nu}$ is Lorentzian, every vector that is causal w.r.t. $m_{\mu\nu}$ is also causal w.r.t. $g_{\mu\nu}$ (so the null cone of $g_{\mu\nu}$ lies on, or outside, the null cone of $m_{\mu\nu}$). 

Most equations are neither subluminal nor superluminal e.g. because the null cones of $g_{\mu\nu}$ and $m_{\mu\nu}$ may not be nested or because the relation between the null cones of $g_{\mu\nu}$ and $m_{\mu\nu}$ may be different for different field configurations. Note also that the standard wave equation ($g_{\mu\nu} = m_{\mu\nu}$) is both subluminal and superluminal according to our definitions. 

Clearly these definitions depends on the choice of $m_{\mu\nu}$. There are infinitely many Minkowski metrics on $\mathbb{R}^{d+1}$. An equation might be subluminal w.r.t. one choice of $m_{\mu\nu}$ and superluminal w.r.t. some other choice. However, for many equations there exists no $m_{\mu\nu}$ such that the equation is either subluminal or superluminal. In physics applications one usually has a preferred choice of $m_{\mu\nu}$, i.e., $m_{\mu\nu}$ is ``the" spacetime metric. In particular, this is the case for the class of Lorentz invariant equations (defined below). 

Since $M$ is a subset of $\mathbb{R}^{d+1}$ it follows that $M$ is orientable because an orientation $(d+1)$-form of $\mathbb{R}^{d+1}$ can be restricted to $M$. In the superluminal case, any vector field $T^\mu$ that is timelike w.r.t. $m_{\mu\nu}$ must also be timelike w.r.t. $g_{\mu\nu}$. It follows that $(M,g)$ is time orientable in the superluminal case. In the subluminal case, note that the null cone of $g^{\mu\nu}$ lies on or outside the null cone of $m^{\mu\nu}$ hence the 1-form $dx^0$ (for inertial frame coordinates $x^\mu$) is timelike w.r.t. $g^{\mu\nu}$. Therefore $T^\mu = -g^{\mu\nu} (dx^0)_\nu = -g^{0 \mu}$ defines a time orientation so $(M,g)$ is time orientable. Furthermore, this shows that $x^0$ is a global time function which implies that $(M,g)$ is stably causal in the subluminal case \cite{mukhanov}.

\subsection{The initial value problem}
\label{SecIVP1}

Let's now discuss the initial value problem for an equation of the form (\ref{eqndims}). Consider prescribing smooth initial data $(S,u,du)$ where $S$ is a hypersurface in $\mathbb{R}^{d+1}$ and $(u,du)$ are specified on $S$. Local well-posedness of the initial value problem requires that initial data is chosen so that $g(u,du)$ is Lorentzian and that $S$ must be spacelike w.r.t. $g(u,du)$. Given such data, one expects a unique hyperbolic solution of (\ref{eqndims}) to exist locally near $S$.\footnote{In fact, for a general equation this is expecting too much. We will discuss this in section \ref{theorems_intro} and Proposition \ref{PropLocUniqueness}.}

We'll say that a hyperbolic solution $(M,u)$ is a {\it development} of the data on $S$ if $S \subset M$ and the solution $(M,u)$ is consistent with the data on $S$. To discuss predictability, we would like to know whether $(M,u)$ is {\it uniquely} determined by the initial data $(S,u,du)$. A necessary condition for such uniqueness is that $(M,g)$ should be globally hyperbolic with Cauchy surface $S$. If $(M,g)$ is not globally hyperbolic then the solution in the region of $M$ beyond the Cauchy horizons $H^\pm(S)$ is not determined uniquely by the data on $S$. We will say that a hyperbolic solution $(M,u)$ is a {\it globally hyperbolic development} (GHD) of the initial data iff $(M,g)$ is globally hyperbolic with Cauchy surface $S$. 

A GHD $(M,u)$ is {\it extendible} if there exists another GHD $(M',u')$ with $M \subsetneq M'$ and $u=u'$ on $M$. We say that $(M,u)$ is a {\it maximal} globally hyperbolic development (MGHD) of the initial data if $(M,u)$ is not extendible as a GHD of the specified data on $S$. Note that a MGHD might be extendible but the extended solution will not be a GHD of the data on $S$: it will exhibit a Cauchy horizon for $S$. 

MGHDs play an important role in General Relativity.  In General Relativity, given initial data for the Einstein equation, there exists a {\it unique} (up to diffeomorphisms) MGHD of the data \cite{cbgeroch}. This MGHD is therefore the central object of interest in GR because it is the largest region of spacetime that can be uniquely predicted from the given initial data. Any well-defined question in the theory can be formulated as a question about the MGHD.\footnote{For example, the strong cosmic censorship conjecture asserts that, for suitable initial data, the MGHD is generically inextendible. The weak cosmic censorship conjecture asserts that, for asymptotically flat initial data, the MGHD generically has a complete future null infinity.}

Surprisingly, the subject of maximal globally hyperbolic developments for equations of the form (\ref{eqndims}) has not received much attention.\footnote{The only exceptions we are aware of are the sketches in \cite{Chr}, Chapter 2, page 40, and in \cite{sbierski}, Section 1.4.1, which  both do not mention the subtleties arising in the case of general wave equations, namely that for two GHDs $u_1 : U_1 \to \R$ and $u_2 :U_2 \to \R$ of the same initial data posed on a connected hypersurface we do not need to have that $U_1 \cap U_2$ is connected. For more on this see our detailed discussion in Section \ref{SecUniqueness}.} 
By analogy with the Einstein equation one might expect a unique MGHD for such an equation. We will see that this is indeed the case for superluminal equations but it is not always true for subluminal equations. The reason that this does not occur for the Einstein equation is that solving the Einstein equation involves constructing the background manifold (which gives flexibility) whereas in solving (\ref{eqndims}) the background manifold is fixed. It is this rigidity which leads to non-uniqueness of MGHDs for subluminal equations.

\section{Born-Infeld scalar in two dimensions}

\label{BI}

\subsection{Two dimensions}

\label{sec:subsupermap}

Let's now consider {\it Lorentz invariant} equations. By this we mean that we pick a Minkowski metric $m_{\mu\nu}$ on $\mathbb{R}^{d+1}$, with constant components in the canonical coordinates $x^\mu$, and we demand that isometries of $m_{\mu\nu}$ map solutions of the equation to solutions of the equation. We will assume that our equation has the form (\ref{eqndims}) where now $g=g(m,u,du)$ and $F=F(m,u,du)$ depend on the choice of $m$. 

The two-dimensional case is special because if $m$ is a Minkowski metric then so is 
\be 
\hat{m} = -m
\ee
Using this fact we can relate subluminal and superluminal equations. Define
\be
 \hat{g} (m,u,du) = -g(-m,u,du)
\ee
and 
\be
 \hat{F}(m,u,du) = -F(-m,u,du)
\ee
Now $u$ satisfies (\ref{eqndims}) if, and only if, it satisfies
\be
\label{hattedeq}
 \hat{g}^{\mu\nu}(\hat{m}, u,du) \partial_\mu \partial_\nu u = \hat{F}(\hat{m}, u,du)
\ee
We view this equation as describing a scalar field in 2d Minkowski spacetime with metric $\hat{m}$. It is easy to see that if (\ref{eqndims}) is a subluminal equation then (\ref{hattedeq}) is superluminal, and vice-versa. 

Since the above transformation reverses the overal sign of $m$ and $g$, it maps timelike vectors to spacelike vectors and vice-versa, i.e., the causal ``cones" of the two theories are the complements of each other. This means that any solution of a superluminal equation arises from a solution of the corresponding subluminal equation simply by interchanging the definitions of timelike and spacelike. For example, if one draws a spacetime diagram for a solution of the subluminal equation, with time running from bottom to top, then the same diagram describes a solution of the superluminal equation, with time running from left to right (or right to left: one still has the freedom to choose the time orientation).

In the Appendix we discuss some general properties of superluminal equations in two dimensions, in particular the question of whether solutions of such an equation can exhibit ``causality violation". 

\subsection{Born-Infeld scalar}

In two dimensional Minkowski spacetime, consider a scalar field with equation of motion obtained from the Born-Infeld action
\be
\label{BIaction}
 S =-\frac{1}{c} \int d^2 x \sqrt{ 1 + c m^{\mu\nu} \partial_\mu \Phi \partial_\nu \Phi }
\ee
where $c$ is a constant. By rescaling the coordinates we can set $c=\pm 1$. The case $c=1$ is the standard Born-Infeld theory. This theory is referred to as ``exceptional" because, unlike in most nonlinear theories, a wavepacket in this theory propagates without distortion and never forms a shock \cite{taniuti}. 

The equation of motion is
\be
\label{BIeom}
 g^{\mu\nu} \partial_\mu \partial_\nu \Phi=0
\ee
where
\be
 g^{\mu\nu} = m^{\mu\nu} - \frac{c m^{\mu\rho}m^{\nu\sigma} \partial_\rho\Phi \partial_\sigma\Phi}{\left(1 + cm^{\lambda \tau} \partial_\lambda \Phi \partial_\tau \Phi \right)}
\ee
The inverse of $g^{\mu\nu}$ is
\be
\label{gdef}
 g_{\mu\nu} = m_{\mu\nu} + c\partial_\mu \Phi \partial_\nu \Phi
\ee
A calculation gives
\be
 \det g_{\mu\nu} = - \left(1 +c m^{\rho\sigma} \partial_\rho \Phi \partial_\sigma \Phi \right)
\ee
Hence $g$ is a Lorentzian metric (i.e. the equation of motion is hyperbolic) if, and only if,
\be
\label{lorentzian}
 1 +c m^{\rho\sigma} \partial_\rho \Phi \partial_\sigma \Phi >0
\ee
In the language of section \ref{sec:subsuper}, a hyperbolic solution must satisfy this inequality. 

Consider a vector $V^\mu$. Note that
\be
  m_{\mu\nu} V^\mu V^\nu = g_{\mu\nu}V^\mu V^\nu - c \left( V \cdot \partial \Phi \right)^2
\ee
If $c=1$ then the final term is non-positive. Hence if $V$ is causal w.r.t. $g_{\mu\nu}$ then $V$ is causal w.r.t. $m_{\mu\nu}$, i.e., the null cone of $g$ lies on or inside that of $m$. However, for $c=-1$, the null cone of $m$ lies on or inside that of $g$. Hence the $c=+1$ theory is subluminal and the $c=-1$ is superluminal according to the definitions of section \ref{sec:subsuper}. 

The two theories are related by the transformation $(c,m,g) \rightarrow (-c,-m,-g)$ with $\Phi$ fixed. This is the map described in section \ref{sec:subsupermap}. 

\subsection{Relation to Nambu-Goto string}

It is well-known that the $c=1$ theory is a gauge-fixed version of an infinite Nambu-Goto string whose target space is $2+1$ dimensional Minkowski spacetime. The same is true for $c=-1$ except that the target space now has $+--$ signature, i.e., two time dimensions. The action of such a string is
\be
 S_{NG} = - \int d^2 x \sqrt{ - \det g}
\ee
where
\be
 g_{\mu\nu} = G_{AB} \partial_\mu X^A \partial_\nu X^B
\ee
with $G_{AB} = {\rm diag}(-1,1,c)$ ($c=\pm 1$), $x^\mu$ are worldsheet coordinates, and $X^A(x)$ are the embedding coordinates of the string. It is assumed that the worldsheet of the string is timelike, i.e., that $g_{\mu\nu}$ has Lorentzian signature.
Fixing the gauge as
\be
\label{gauge}
 x^0=X^0 \qquad x^1 = X^1
\ee
and defining $\Phi(x) = X^2(x)$, the action reduces to that of the Born-Infeld scalar described above, and the worldsheet metric $g_{\mu\nu}$ is the same as the effective metric given by equation (\ref{gdef}). Note that the $c=\pm 1$ theories are mapped to each other under the transformation $(G,g) \rightarrow (-G,-g)$. From the worldsheet point of view, this corresponds to interchanging the definitions of timelike and spacelike, as discussed above. 

Although the Born-Infeld scalar can be obtained from the Nambu-Goto string, we will not regard them as equivalent theories. We will view the BI scalar as a theory defined in a global 2-dimensional Minkowski spacetime. No such spacetime is present for the Nambu-Goto string. Of course any solution of the BI scalar theory can be ``uplifted" to give some solution for the Nambu-Goto string. However, the converse is not true because not all solutions of the Nambu-Goto string can be written in the gauge (\ref{gauge}). In particular, string profiles which ``fold back" on themselves as in Fig. \ref{fig:ngstringprofile} are excluded by this gauge choice. From the BI perspective, such configurations will look singular. Of course such singularities can be eliminated by returning to the Nambu-Goto picture. However, we will not do this: the point is that the BI scalar is our guide to possible behaviour of nonlinear scalar field theories in 2d Minkowski spacetime, and most such theories do not have any analogue of the Nambu-Goto string interpretation. 

\subsection{Non-uniqueness}

\label{sec:nonuniqueness}

We can use the Nambu-Goto string to explain heuristically why there is a problem with the {\it subluminal} Born-Infeld scalar theory. (The superluminal case is harder to discuss heuristically because in this case the Nambu-Goto target space has two time directions.) Consider a left moving and a right moving wavepacket propagating along the string. As we will review below, if the wavepackets are sufficiently strong, when they intersect then the string can fold back on itself as described above. This is shown in Fig. \ref{fig:ngstringprofile}. 
When this happens, the field $\Phi$ ``wants to become multi-valued". But this is not possible in the BI theory because $\Phi$ is a scalar field in 2d Minkowski spacetime so $\Phi$ must be single-valued. 

\begin{figure}
	\centering
	\includegraphics[width=0.5\linewidth]{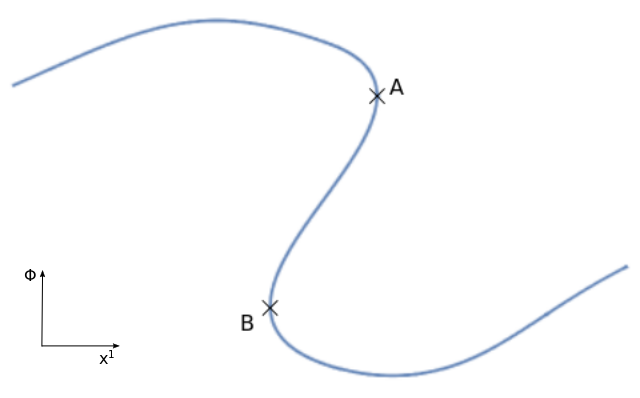}
	\caption{An example of the string folding back on itself. The gradient is infinite at points $A$ and $B$.}
	\label{fig:ngstringprofile}
\end{figure}

Clearly we have to ``choose a branch" of the solution $\Phi$ at each point of 2d Minkowsi spacetime. We want to do this so that the solution is as smooth as possible. There are two obvious ways of doing this. We could start from the left of the string and extend until we reach the point $A$ of infinite gradient as shown in Fig \ref{fig:ngstringprofile}. But beyond this point we have to jump to the other branch, so the solution is discontinuous as shown in Fig. \ref{fig:ngstringprofilea}. If the discontinuity is approached from the left then the gradient of $\Phi$ diverges as we approach $A$. However, if approached from the right the gradient remains bounded up to the discontinuity at $A$. Following out this procedure for the full spacetime produces a globally defined solution of the Born-Infeld theory. After some time, the wavepackets on the Nambu-Goto string separate and the resulting Born-Infeld solution becomes continuous again. 

Now note that instead of starting on the left and extending to point $A$ we could have started on the right and extended to point $B$. Now the discontinuity would occur at $B$ instead of $A$. So now the solution appears as shown in Fig. \ref{fig:ngstringprofileb}. Approaching the discontinuity from the right, the gradient of $\Phi$ diverges at $B$. However approaching from the left, the gradient remains bounded up to the discontinuity at $B$. As above, this procedure gives a globally defined solution of the Born-Infeld theory. This is clearly a {\it different} solution from the solution discussed in the previous paragraph. 

\begin{figure}[h!]
	\centering
	\begin{minipage}{0.45\textwidth}
		\centering
		\includegraphics[width=0.8\linewidth]{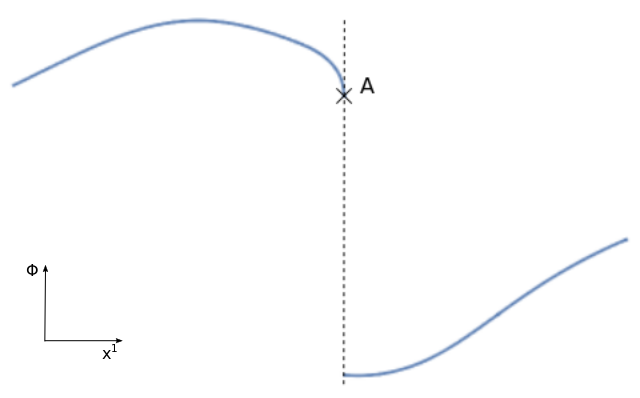}
		\caption{Solution with discontinuity at $A$.}
		\label{fig:ngstringprofilea}
		
	\end{minipage}\hfill
	\begin{minipage}{0.45\textwidth}
		\centering
		\includegraphics[width=0.8\linewidth]{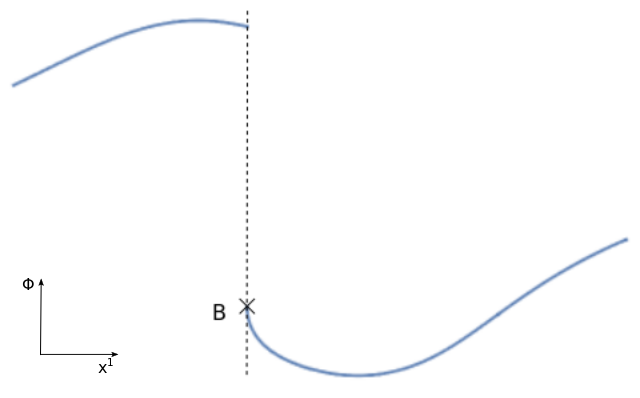}
		\caption{Different solution with discontinuity at $B$.}
		\label{fig:ngstringprofileb}
	\end{minipage}
\end{figure}

Starting from initial data prescribed on some line $S$ in the far past, the above constructions produce two {\it different} solutions which agree with the data on $S$. Now non-uniqueness is to be expected because the solution $\Phi$ is singular (at $A$ or $B$), so the corresponding spacetimes $(M,g)$ will not be globally hyperbolic. Therefore lack of uniqueness is to be expected beyond the Cauchy horizon. However, we will show, in the subluminal case, that the lack of uniqueness {\it occurs before a Cauchy horizon forms}. In other words, the two solutions disagree in a region which belongs to $D^+(S)$ for both solutions. This implies that the two solutions cannot arise from the same MGHD of the data on $S$. Therefore MGHDs are not unique.  

Clearly there are other ways we could construct Born-Infeld solutions from the Nambu-Goto solution: we do not have to take the discontinuity to occur at either point $A$ or at point $B$, we could take it to occur at any point between $A$ and $B$. This leads to an infinite set of possible solutions, and an infinite set of distinct MGHDs.

The above discussion was for the subluminal ($c=1$) theory. We will show below that this problem does not occur for the superluminal theory. This is because, in the superluminal theory, from the 2d Born-Infeld perspective, $A$ and $B$ are {\it timelike} separated with $B$ (say) occuring to the future of $A$. This implies that $B$ lies to the future of the infinite gradient singularity at $A$ hence $B$ cannot belong to $D^+(S)$ if $S$ is a surface to the past of $A$. Therefore there is a unique choice of branch in the superluminal theory. In this theory there is a unique MGHD. 
 
\subsection{General solution}

\label{sec:gensol}

The $c=1$ (subluminal) BI scalar theory was solved by Barbashov and Chernikov \cite{Barbashov:1966frq,Barbashov:1967avx}. We will follow the notation of Whitham \cite{whitham}, who gives a nice summary of their work. Because the superluminal and subluminal theories are related as discussed above, it is easy to write down the general hyperbolic\footnote{This solution was obtained using the method of characteristics which only works when the equation is hyperbolic so only hyperbolic solutions are obtained using this method.} solution for both cases. 
Write the Minkowski metric as
\be
\label{minkbi}
m=-c (dx^0)^2 +c (dx^1)^2
\ee
and define null coordinates
\be
 \xi = x^1-x^0 \qquad \eta = x^1+x^0
\ee
The solution is written in terms of a mapping $\Psi : \mathbb{R}^2 \rightarrow \mathbb {R}^2$ given by
\be
 \Psi : (\rho,\sigma) \mapsto (\xi(\rho,\sigma),\eta(\rho,\sigma))
\ee
where
\be
\label{xisol}
 \xi(\rho,\sigma) = \rho - \int_{-\infty}^\sigma \Phi_2'(x)^2 dx
 \ee
 and
 \be
\label{etasol}
\eta(\rho,\sigma)= \sigma + \int_\rho^\infty \Phi_1'(x)^2 dx
\ee
with $\Phi_1(\rho)$ and $\Phi_2(\sigma)$ smooth functions such that $\Phi_1'$ and $\Phi_2'$ decay at infinity fast enough to ensure that the integrals converge.\footnote{The latter assumption could be relaxed by replacing the infinite limits of the integrals by finite constants.} These two functions can be viewed as specifying the profiles of left moving and right moving wavepackets.

Assuming that $\Psi$ is invertible we can write $\rho = \rho(\xi,\eta)$ and $\sigma=\sigma(\xi,\eta)$ and the solution is given by
\be
\label{Phi_sol}
 \Phi(\xi,\eta) = \Phi_1(\rho(\xi,\eta)) + \Phi_2(\sigma(\xi,\eta)) 
\ee 
We can state the above result as a theorem \cite{Barbashov:1966frq,Barbashov:1967avx,whitham}:
\begin{theorem}
Let $\Phi_1(\rho)$ and $\Phi_2(\sigma)$ be smooth functions defined for all $(\rho,\sigma) \in \R^2$. Let $V$ be a connected open subset of $\mathbb{R}^2$. If the map $\Psi:V \rightarrow U \subset \mathbb{R}^2$ defined by (\ref{xisol}) and (\ref{etasol}) is a diffeomorphism then (\ref{Phi_sol}) defines a smooth solution $\Phi: U \rightarrow  \mathbb{R}$ of the Born-Infeld scalar equation of motion. 
\label{diffeo_theorem}
\end{theorem}

Clearly it will be important to determine whether or not $\Psi$ is a diffeomorphism.  
 \begin{lemma}
 A necessary (although not sufficient) condition for $\Psi:V \rightarrow U$ to be a diffeomorphism is that either $ \Phi_1'(\rho)^2 \Phi_2'(\sigma)^2 <1$ throughout $V$ or $ \Phi_1'(\rho)^2 \Phi_2'(\sigma)^2 >1$ throughout $V$. 
\label{nec_cond}
\end{lemma}
\begin{proof}
The Jacobian of the map $\Psi$ is
\be
\label{jacobian}
 \det \frac{\partial(\xi,\eta)}{\partial (\rho,\sigma)} = 1-\Phi_1'(\rho)^2 \Phi_2'(\sigma)^2
\ee 
hence a necessary condition for $\Psi$ to define a diffeomorphism is that the RHS cannot vanish at any point of $V$. Since $V$ is connected the result follows immediately.
\end{proof}

A point on the boundary $\partial V$ at which $\Phi_1' \Phi_2' = 1$ corresponds to a singularity:

\begin{lemma}
Assume that $\Psi:V\rightarrow U$ is a diffeomorphism such that $\Phi_1'(\rho) \Phi_2'(\sigma) \rightarrow 1$ as $(\rho,\sigma) \rightarrow (\rho_0,\sigma_0)$ for some $(\rho_0,\sigma_0) \in \partial V$. Let $\gamma:(0,1) \rightarrow V$ be a smooth curve with $\gamma(t) \rightarrow (\rho_0,\sigma_0)$ as $t \rightarrow 1$. Then the gradient of the solution $\Phi$ at the point $\Psi(\gamma(t))$ diverges as $t \rightarrow 1$. 
\label{lemma_sing}
\end{lemma}
\begin{proof}
A calculation gives 
\be \label{NullDerivativesRhoSigma}
 \partial_\xi \Phi = \frac{\Phi_1'(\rho)}{1-\Phi_1'(\rho) \Phi_2'(\sigma)} \qquad \partial_\eta\Phi = \frac{\Phi_2'(\sigma)}{1-\Phi_1'(\rho) \Phi_2'(\sigma)}
\ee
The result follows immediately.
\end{proof}

It can be shown similarly that points of $\partial V$ where $\Phi_1' \Phi_2' = -1$ correspond to a divergence in the second derivative of $\Phi$ although we will not need this result below.

We will be mainly interested in causal properties of the metric $g$ defined by (\ref{gdef}). If $\Psi:V \rightarrow U$ is a diffeomorphism then we can introduce $(\rho,\sigma)$ as coordinates on $V$. The metric $g$ defined by (\ref{gdef}) takes a simple form in these coordinates:

\begin{lemma}
Consider a Born-Infeld solution constructed as in Theorem \ref{diffeo_theorem}. In coordinates $(\rho,\sigma)$, the metric (\ref{gdef}) is
\be
\label{grhosigma}
 g =c  \left( 1 + \Phi_1'(\rho) \Phi_2'(\sigma) \right)^2 d\rho d\sigma
\ee
\label{glemma}
\end{lemma}
\begin{proof}
Direct calculation using (\ref{xisol}), (\ref{etasol}) and (\ref{Phi_sol}). 
\end{proof}

Note that the vector fields $\partial/\partial \rho$ and $\partial/\partial \sigma$ are null w.r.t. $g$. Let's determine whether they are future or past directed. Recall (section \ref{sec:subsuper}) that the time-orientation for $g$ is determined by a choice of time orientation for Minkowski spacetime. 

\begin{lemma}
Consider a Born-Infeld solution constructed as in Theorem \ref{diffeo_theorem}. In the subluminal case, $\partial/\partial \rho$ is past-directed and $\partial/\partial \sigma$ is future-directed w.r.t. $g$. In the superluminal case, if $\Phi_1'(\rho)^2 \Phi_2'(\sigma)^2 < 1$ then $\partial/\partial \rho$ and $\partial/\partial \sigma$ are both future directed whereas if $\Phi_1'(\rho)^2 \Phi_2'(\sigma)^2 > 1$ then they are both past-directed. In either case the spacetime $(U,g)$ is stably causal.
\label{lemma_orientation}
\end{lemma}
\begin{proof}
In the subluminal case ($c=1$) we know (section \ref{sec:subsuper}) that $x^0$ is a global time function for the spacetime $(U,g)$ so this spacetime is stably causal. From (\ref{xisol}) and (\ref{etasol}) one finds $\partial x^0/\partial \rho<0$ and $\partial x^0/\partial \sigma>0$ and the result follows.

In the superluminal case ($c=-1$), $\partial/\partial x^1$ is timelike w.r.t. $m$ so (section \ref{sec:subsuper}) we choose $\partial/\partial x^1$ as a time-orientation on $(V,g)$. A calculation gives
\be
 \frac{\partial}{\partial x^1} = \frac{1}{1-\Phi_1'(\rho)^2 \Phi_2'(\sigma)^2} \left[ ( 1+ \Phi_2'(\sigma)^2) \frac{\partial}{\partial \rho} + ( 1+ \Phi_1'(\rho)^2) \frac{\partial}{\partial \sigma}  \right] 
\ee
The inner products (w.r.t. $g$) of $\partial /\partial x^1$ with $\partial/\partial \rho$ and $\partial/\partial \sigma$ can be calculated using (\ref{grhosigma}). Clearly these inner products have the opposite sign to $1-\Phi_1'(\rho)^2 \Phi_2'(\sigma)^2$ and so $\partial/\partial \rho$ and $\partial/\partial \sigma$ are both future directed if this quantity is positive and past directed if it is negative. If $\Phi_1'(\rho)^2 \Phi_2'(\sigma)^2<1$ then let $X = \partial/\partial \rho + \partial/\partial \sigma$, which is future-directed and timelike w.r.t. $g$. We then have $g_{\mu\nu} X^\nu \propto -[d (\rho+\sigma)]_\mu$ hence $\rho+\sigma$ is a global time function for $(U,g)$ and so $(U,g)$ is stably causal. Similarly if $\Phi_1'(\rho)^2 \Phi_2'(\sigma)^2 > 1$  then $-(\rho+\sigma)$ is a global time function for $(U,g)$. 
 \end{proof}

In the superluminal case, this proves that solutions constructed using Theorem \ref{diffeo_theorem} cannot exhibit any violation of causality. However, we note that there may be solutions of (\ref{BIeom}) that cannot be obtained using Theorem \ref{diffeo_theorem}. Such solutions would requires multiple charts $V_\alpha$, each with corresponding coordinates $(\rho_\alpha,\sigma_\alpha)$ and diffeomorphisms $\Psi_\alpha$. In any given chart the solution will take the form described above. With multiple charts, it may not be possible to construct a global time function for the superluminal theory. 

We are interested in globally hyperbolic developments of initial data. It is very easy to determine whether or not a solution constructed using Theorem \ref{diffeo_theorem} is globally hyperbolic:

\begin{lemma}
Consider a Born-Infeld solution constructed as in Theorem \ref{diffeo_theorem}. Then $(U,g)$ is globally hyperbolic with Cauchy surface $S$ if, and only if, $(V,\hat{m})$ is globally hyperbolic with Cauchy surface $\Sigma=\Psi^{-1}(S)$, where $\hat{m} = c d\rho d\sigma$. 
\label{glob_hyp_lemma}
\end{lemma}
\begin{proof}
This is an immediate consequence of (\ref{grhosigma}) which shows that $g$ and $\hat{m}$ define causally equivalent metrics on $V$. (Here we are not bothering to distinguish the metric $g$ on $U$ and the metric on $V$ defined by pull-back of $g$ w.r.t. $\Psi$.)
\end{proof}

Thus global hyperbolicity can be checked using the flat metric $\hat{m}$ on $V$. More generally, the causal properties of $(U,g)$ are the same as those of the flat spacetime $(V,\hat{m})$.

We will show that, given initial data on a surface $S$, there exist multiple distinct maximal globally hyperbolic developments in the subluminal case ($c=1$) but there is a unique MGHD in the superluminal ($c=-1$) case. This difference can be traced to the following property:

\begin{lemma}
Let $p,\,q$ be distinct points such that $\Psi(p)=\Psi(q)$. Then the straight line connecting $p,q$ in the $(\rho,\sigma)$ plane is spacelike w.r.t. $\hat{m}$ in the subluminal case and timelike in the superluminal case. 
\label{sep_lemma}
\end{lemma}
\begin{proof}
Let $p$ and $q$ have coordinates $(\rho_2,\,\sigma_2)$ and  $(\rho_1,\,\sigma_1)$ respectively.
From equations (\ref{xisol}) and (\ref{etasol}) we have
\begin{equation}
\delta \rho \equiv \rho_2-\rho_1=\int_{\sigma_1}^{\sigma_2}\Phi_2'(x)^2 dx, \qquad \delta \sigma \equiv \sigma_2-\sigma_1=\int_{\rho_1}^{\rho_2}\Phi_1'(x)^2 dx. \label{eq:difference}
\end{equation}
From the first equation we see that $\delta \sigma = 0$ implies $\delta \rho=0$ and the second equation gives the converse. Hence $\delta \rho=0$ if, and only if, $\delta \sigma=0$, i.e., $p=q$. Since we are assuming $p \ne q$ we must have $\delta \rho \ne 0$ and $\delta \sigma \ne 0$. The first equation then implies that $\delta \rho$ has the same sign as $\delta \sigma$ so
\begin{equation}
\label{deltarhodeltasigma}
 \delta \rho\, \delta \sigma >0.
\end{equation}
The result follows from the definition of $\hat{m}$ in Lemma \ref{glob_hyp_lemma}. 
\end{proof}

Theorem \ref{diffeo_theorem} defines a solution in a subset $U$ of Minkowski spacetime. The following theorem \cite{Barbashov:1967avx} guarantees a {\it global} solution:

\begin{theorem}
Let $\Phi_1$ and $\Phi_2$ be smooth functions on the real line such that the integrals in (\ref{xisol}) and (\ref{etasol}) converge for $\rho \rightarrow -\infty$ and $\sigma \rightarrow \infty$. Assume that  $\Phi_1'(\rho)^2 \Phi_2'(\sigma)^2 < 1$ for all $(\rho,\sigma)$. Then the map $\Psi: \mathbb{R}^2 \rightarrow \mathbb{R}^2$ defined by (\ref{xisol}), (\ref{etasol}) is a diffeomorphism and so the Born-Infeld solution of Theorem \ref{diffeo_theorem} is a globally defined smooth solution.
\label{global_theorem}
\end{theorem}
\begin{proof}
Following \cite{Barbashov:1967avx}, use (\ref{etasol}) to write
\be
\label{sigmasol}
 \sigma = \sigma_\eta(\rho) \equiv \eta - \int_\rho^\infty \Phi_1'(x)^2 dx
 \ee
and then substitute into (\ref{xisol}) to obtain
\be
\label{rhoeq}
 \xi = F(\rho;\eta) \equiv  \rho - \int_{-\infty}^{\sigma_\eta(\rho)} \Phi_2'(x)^2dx
\ee
We want to use this equation to determine $\rho$ as a function of $\xi,\eta$. A calculation gives
\be
\label{dFdrho}
 \left( \frac{\partial F}{\partial \rho} \right)_\eta = 1 - \Phi_1'(\rho)^2 \Phi_2'(\sigma_\eta(\rho))^2
\ee
So $\Phi_1'(\rho)^2 \Phi_2'(\sigma)^2 < 1$ implies that $F$ is a strictly increasing function of $\rho$ and hence there exists at most one solution $\rho$ of (\ref{rhoeq}) for any $(\xi,\eta)$. Given a solution for $\rho$, (\ref{sigmasol}) determines $\sigma$ uniquely. This proves that the map $\Psi$ is injective. 

We now show that there exists exactly one solution of (\ref{rhoeq}). Our assumptions on $\Phi_1$ imply that $\sigma_\eta(\rho)\rightarrow \eta$ as $\rho \rightarrow \infty$ and $\sigma_\eta(\rho) \rightarrow \eta-C$ as $\rho \rightarrow - \infty$ where $C = \int_{-\infty}^\infty \Phi_1'(x)^2 dx$. Our assumptions on $\Phi_1$ imply that $\Phi_1'(\rho) \rightarrow 0$ as $\rho \rightarrow \pm \infty$. So now from (\ref{dFdrho}) we see that $(\partial F/\partial \rho)_\eta \rightarrow 1$ as $ \rho \rightarrow \pm \infty$. So, at fixed $\eta$, $F$ is strictly increasing and has gradient $1$ for $\rho \rightarrow \pm \infty$. This implies that, at fixed $\eta$, the map $\rho \rightarrow F(\rho;\eta)$ is a bijection from $\mathbb{R}$ to itself. Hence there exists exactly one solution of (\ref{rhoeq}) for given $(\xi,\eta)$. Hence $\Psi$ is a bijection. That $\Psi$ is a diffeomorphism now follows from the fact that the RHS of equation (\ref{jacobian}) is everywhere non-zero.
\end{proof}

\begin{lemma}
The solution of Theorem \ref{global_theorem} is globally hyperbolic.
\end{lemma}
\begin{proof}
This follows immediately from Lemma \ref{glob_hyp_lemma} because $(V,\hat{m})=(\mathbb{R}^2,\hat{m})$ so the causal structure w.r.t. $g$ is the same as 2d Minkowski spacetime. In the subluminal case, surfaces of constant $x^0$ are Cauchy because $x^0$ is a global time function. In the superluminal case, a surface of constant $\rho + \sigma$ is Cauchy since the proof of Lemma \ref{lemma_orientation} shows that $\rho+\sigma$ is a global time function. 
\end{proof}

As discussed above, we need $\Psi$ to be a diffeomorphism for equations \eqref{xisol}, \eqref{etasol}, \eqref{Phi_sol} to define a solution of the Born-Infeld scalar. However, we note that these equations define a solution of the Nambu-Goto string irrespective of whether or not $\Psi$ is a diffeomorphism. To see this, take $(\rho,\sigma)$ as worldsheet coordinates and replace the LHS of (\ref{xisol}) and (\ref{etasol}) by $X^1-X^0$ and $X^1+X^0$ respectively. Together with $X^2=\Phi=\Phi_1(\rho)+\Phi_2(\sigma)$ this specifies a globally well defined embedding of the string worldsheet into $\mathbb{R}^3$. The worldsheet metric is (\ref{grhosigma}). The solution describes a superposition of left moving and right moving wavepackets described by $\Phi_1(\rho)$ and $\Phi_2(\sigma)$, each travelling at the speed of light with respect to $g$. The worldsheet metric degenerates at points where $\Phi_1'(\rho) \Phi_2'(\sigma) = - 1$. These correspond to ``cusp" singularities  at which the string worldsheet becomes null. The string is smooth at points where $\Phi_1'(\rho) \Phi_2'(\sigma) = +1$, which correspond to points of infinite gradient like A or B in Fig. \ref{fig:ngstringprofile}. 

\subsection{Example of non-uniqueness in subluminal case}

\label{SecIVP}

We start by recording that the subluminal Born-Infeld scalar equation of motion \eqref{BIeom} written out in coordinates $x^\mu$ reduces to:
\begin{equation} \label{Eqc1}
-\big( 1 + (\partial_{x^1}\Phi)^2\big) \partial_{x^0}^2 \Phi + 2\partial_{x^0}\Phi \partial_{x^1} \Phi \cdot \partial_{x^0}\partial_{x^1} \Phi + \big(1-(\partial_{x^0}\Phi)^2\big)\partial_{x^1}^2 \Phi =0 \;.
\end{equation}
In this section we will demonstrate the existence of two different maximal globally hyperbolic developments (MGHDs) arising from the same initial data for the above equation. We will do this with an example involving a specific choice of the functions $\Phi_1$ and $\Phi_2$, and construct solutions using Theorem \ref{diffeo_theorem}. 

 To construct a solution of \eqref{Eqc1} we choose functions
\be
 \Phi_1(x) = \Phi_2(x) = \phi(x) \equiv \int_{-\infty}^x a e^{-t^2}dt
\label{phi_def}
\ee
where $a>1$ is a constant. This gives $\Phi'_1(x) = \Phi'_2(x) = \phi'(x) :=  a e^{-x^2}$. Hence $\Phi'_1(\rho)^2 \Phi'_2(\sigma)^2 = a^2 e^{-r^2}$ where $r = \sqrt{\rho^2 + \sigma^2}$. Let $r_0 = \sqrt{2 \ln(a)}$. In the $(\rho, \sigma)$ plane we have 
\begin{equation}\label{PhiPrime}
\begin{aligned}
\Phi_1'(\rho)\Phi_2'(\sigma) &< 1 \quad \textnormal{ outside the circle of radius } r_0\\ 
\Phi_1'(\rho)\Phi_2'(\sigma) &> 1 \quad \textnormal{ inside the circle of radius } r_0 \\ 
\Phi_1'(\rho)\Phi_2'(\sigma) &= 1 \quad  \textnormal{ on the circle of radius } r_0\;.
\end{aligned}
\end{equation}
Theorem \ref{global_theorem} does {\it not} apply, and we do not have a global solution. Indeed the map $\Psi$ defined by this choice of $\Phi_1$ and $\Phi_2$ is not injective on $\mathbb{R}^2$. In section \ref{numerics} we will determine numerically the region in which injectivity fails and explain heuristically how this leads to non-uniqueness of MGHDs. Then, in section \ref{nutheorem} we will use the above example to prove a theorem establishing non-uniqueness of MGHDs.

\subsubsection{Numerical demonstration of non-uniqueness of MGHDs}

\label{numerics}

{\it Step 1.} We start by showing that, for the example \eqref{phi_def}, $\Psi$ is non-injective on $\mathbb{R}^2$ but its restriction to a subset $V'$ of $\mathbb{R}^2$ {\it is} injective and so we obtain a solution of \eqref{Eqc1} via Theorem \ref{diffeo_theorem}.
 
The region in which injectivity of $\Psi$ fails can be determined numerically\footnote{These plots were determined using the \textit{FindRoot} function in Mathematica to numerically construct an inverse function. A different starting point for the numerics was used for each region.} and is shown in Fig. \ref{fig:Urhosigma}: three open regions $D$, $E$ and $F$ of the $(\rho,\sigma)$ plane map to the same region $X$ of Minkowski spacetime. Here $D$ is the disc $r<r_0$. The region $X\equiv \Psi(D)$ is shown in Fig. \ref{fig:Uxieta}. The inverse image of any point in $X$ consists of three points, one in each of $D$, $E$ and $F$.\footnote{In the Nambu-Goto string interpretation, $X$ is is the region of spacetime in which the string worldsheet folds back on itself as in Fig. \ref{fig:ngstringprofile}.} However, the map $\Psi$ {\it is} injective on $V' \equiv \mathbb{R}^2  \backslash \overline{D \cup E \cup F}$ and \eqref{PhiPrime} implies that the condition of Lemma \ref{nec_cond} is satisfied on $V'$ so $\Psi$ defines a diffeomorphism from $V'$ to $U' \equiv \Psi(V') = \mathbb{R}^2 \backslash \overline{X}$. Hence Theorem \ref{diffeo_theorem} defines a solution $\Phi: U' \rightarrow \mathbb{R}^2$ of \eqref{Eqc1}. 
\begin{figure}[h!]
	\centering
	\begin{minipage}{0.45\textwidth}
		\centering
		\includegraphics[width=0.95\linewidth]{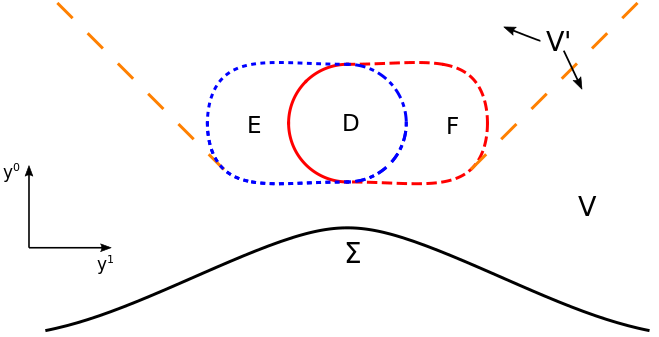}
		\caption{Plot of the $(\rho,\sigma)$ plane in coordinates $(y^0,y^1)$ defined by \eqref{ydef}. The open sets $D,E,F$ have the same image under $\Psi$. The dotted blue (dashed red) curve has the same image as the dot-dashed blue (solid red) curve. 
	 $\Psi$ is injective on $V'$, the complement of $\overline{D \cup E \cup F}$. The orange large dashed lines are the future Cauchy horizon for the initial surface $\Sigma$ in the flat spacetime $(V',\hat{m})$. $V$ is the region of $V'$ lying to the past of this Cauchy horizon.
		}
		\label{fig:Urhosigma}
		
	\end{minipage}\hfill
	\begin{minipage}{0.45\textwidth}
		\centering
		\includegraphics[width=0.75\linewidth]{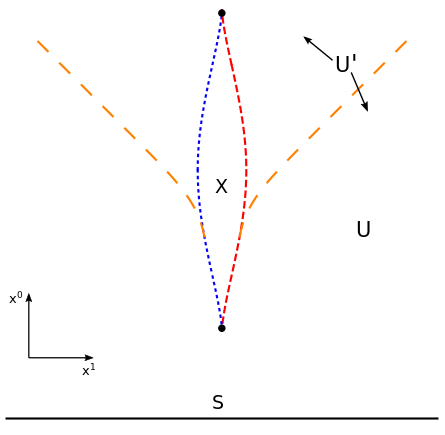}
		\caption{Minkowski spacetime with coordinates $(x^0,x^1)$. The region $X$ is the image of $D$ (or $E$ or $F$) under $\Psi$. The region $U'$ is the complement of $\overline{X}$. The two black dots are points at which the gradient of the solution $\Phi:U' \rightarrow \mathbb{R}$ diverges (by Lemma \ref{lemma_sing}). The orange large dashed lines are the future Cauchy horizon of $S$ in the spacetime $(U',g)$. $U$ is the region of $U'$ lying to the past of this Cauchy horizon. $\Phi:U \rightarrow \mathbb{R}$ is a GHD of the initial data on $S$.
		}
		\label{fig:Uxieta}
	\end{minipage}
\end{figure}

{\it Step 2.} Next we will show that the solution $\Phi: U' \rightarrow \mathbb{R}^2$ is not a GHD but, by restricting its domain, we can construct a GHD. 

Lemma \ref{glob_hyp_lemma} establishes that $(U',g)$ is globally hyperbolic if, and only if, $(V',\hat{m})$ is globally hyperbolic. Introduce coordinates $(y^0,y^1)$ in the $(\rho,\sigma)$ plane such that
\be
\label{ydef}
 \rho = y^1 -y^0 \qquad \sigma = y^1 + y^0.
\ee
In these coordinates we have
\be
 \hat{m} = -(dy^0)^2 + (dy^1)^2
\ee
and Lemma \ref{lemma_orientation} implies that $\partial/\partial y^0$ is future-directed. The causal properties of $\hat{m}$ (and hence $g$) in the $(\rho,\sigma)$ plane are easy to read off from Fig. \ref{fig:Urhosigma}. In particular it is clear that the region $V'$ is not globally hyperbolic w.r.t. $\hat{m}$ so $U'$ is not globally hyperbolic w.r.t. $g$. Consider an initial surface $S$ defined by $x^0 = -T$, as shown in Fig. \ref{fig:Uxieta}. Let $U$ be the domain of dependence of $S$ in $(U',g)$. Then by restricting $\Phi$ to $U$ we obtain a GHD $\Phi: U \rightarrow \mathbb{R}$ of the initial data on $S$. Appealing to Lemma \ref{glob_hyp_lemma}, $U=\Psi(V)$ where $V$ is the domain of dependence of $\Sigma \equiv \Psi^{-1}(S)$ in $(V',\hat{m})$. Viewed as a subset of $V'$, $V$ is bounded by the future Cauchy horizon shown in Fig. \ref{fig:Urhosigma}, which maps to a corresponding future Cauchy horizon in Fig. \ref{fig:Uxieta}. 

{\it Step 3.} Now we will show that the GHD $\Phi: U \rightarrow \mathbb{R}$ is not maximal and it can be smoothly extended to give a GHD $\Phi_a: U_a \rightarrow \mathbb{R}$ that contains part of region $X$. We will show that this extended GHD is smooth on the ``left" boundary of $X$ but singular on the ``right" boundary of $X$. 

We enlarge the GHD $\Phi: U \rightarrow \mathbb{R}$ by pushing the left large dashed orange line of Fig. \ref{fig:Urhosigma} into region $E$ until it is tangent to the boundary of $D$. Specifically, consider the region $V_a$ defined in Fig. \ref{fig:Erhosigma}. Since $V_a$ contains no points of $D$ or $F$, the map $\Psi$ is still injective on this enlarged region and still satisfies \eqref{jacobian}, hence $\Psi$ is a diffeomorphism and so Theorem \ref{diffeo_theorem} defines a solution $\Phi_a: U_a \rightarrow \mathbb{R}$ where $U_a = \Psi(V_a)$. Furthermore, $(V_a,\hat{m})$ is globally hyperbolic with Cauchy surface $\Sigma$ and so $(U_a,g)$ is globally hyperbolic with Cauchy surface $S$. Hence $\Phi_a$ is a GHD of the initial data on $S$. The region $U_a$ is shown in Fig. \ref{fig:Exieta}: it extends across the left boundary of $X$ all the way to the right boundary of $X$. This right boundary is not part of $U_a$, indeed the solution $\Phi_a$ is discontinuous across this boundary.\footnote{In the Nambu-Goto string interpretation, the string worldsheet on a surface of constant $x^0$ intersecting $X$ resembles Fig. \ref{fig:ngstringprofilea} with point $A$ on the right boundary of $X$.}

Consider a curve $\gamma_1$ approaching the right boundary of $X$ from the left (i.e. from within $X$) as in Fig. \ref{fig:Exieta}. Then $\Psi^{-1}(\gamma_1)$ is a curve approaching the solid red curve of Fig. \ref{fig:Erhosigma} from within $E$. Since $\Phi_1'\Phi_2'=1$ on this red curve, Lemma \ref{lemma_sing} implies that the gradient of $\Phi_a$ diverges along $\gamma_1$ as one approaches the boundary. Thus the gradient of $\Phi_a$ diverges along the right boundary of $X$ when approached from the left. On the other hand, if $\gamma_2$ is a curve approaching this boundary from the right (i.e. from outside $X$) as in Fig. \ref{fig:Exieta} then $\Psi^{-1}(\gamma_2)$ approaches the dotted red curve of Fig. \ref{fig:Erhosigma}, which is in the region where $\Phi_1'\Phi_2'<1$ so the gradient of $\Phi_a$ remains bounded along $\gamma_2$. Hence the gradient of $\Phi_a$ is bounded as one approaches the right boundary of $X$ from outside $X$.\begin{figure}[h!]
	\centering
	\begin{minipage}{0.45\textwidth}
		\centering
		\includegraphics[width=0.95\linewidth]{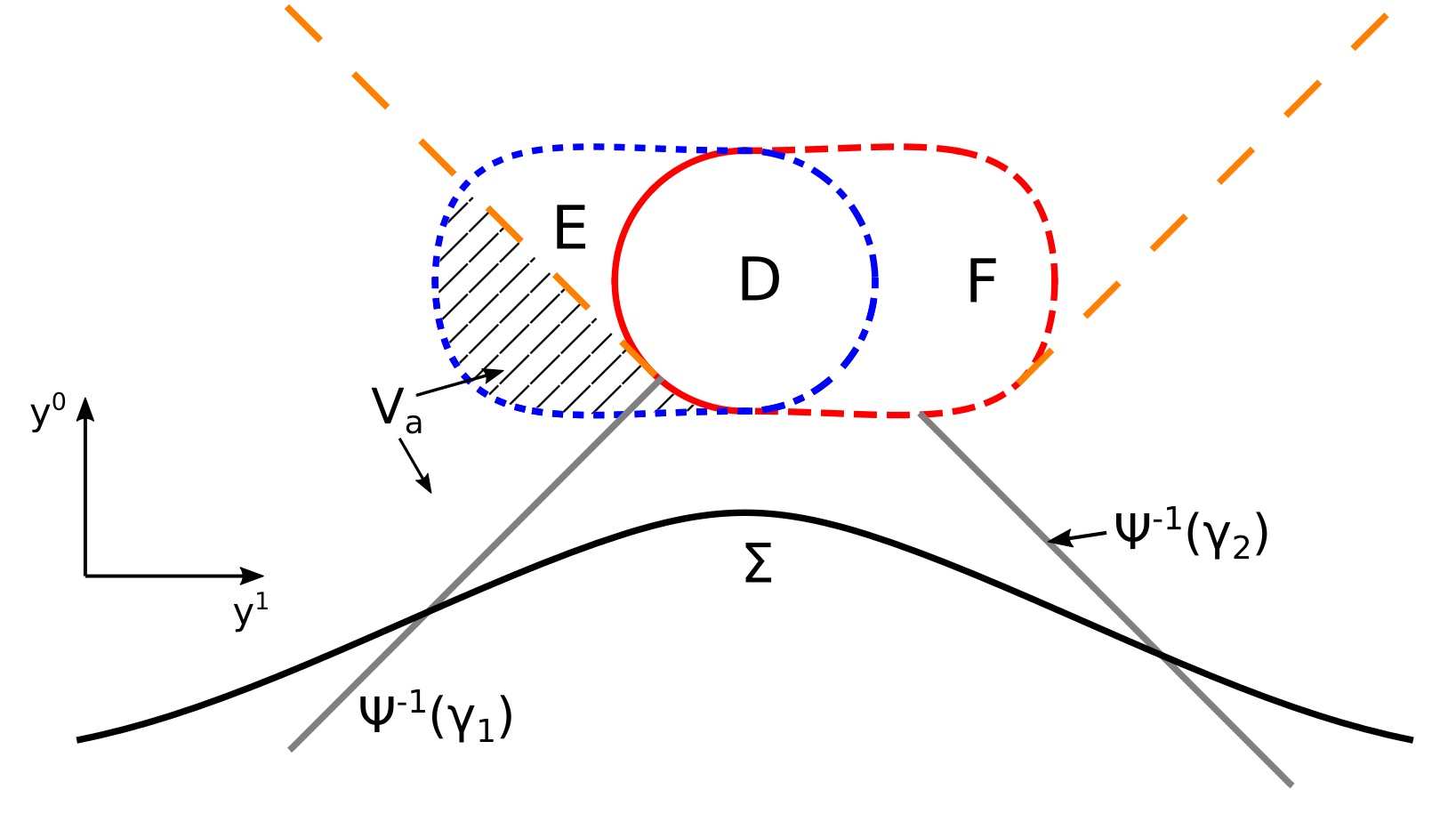}
		\caption{The large dashed orange line on the left is a line of constant $\sigma$ which is tangent to the boundary of $D$ at their point of contact. The region $V_a$ is the union of $V$ with the region to the past (w.r.t. $\hat{m}$) of this line and the shaded section of $E$. The future bounday of $V_a$ consists of the pair of large dashed orange null lines together with the (spacelike) sections of the solid and dashed red curves that connect them.}
		\label{fig:Erhosigma}
		
	\end{minipage}\hfill
	\begin{minipage}{0.45\textwidth}
		\centering
		\includegraphics[width=0.85\linewidth]{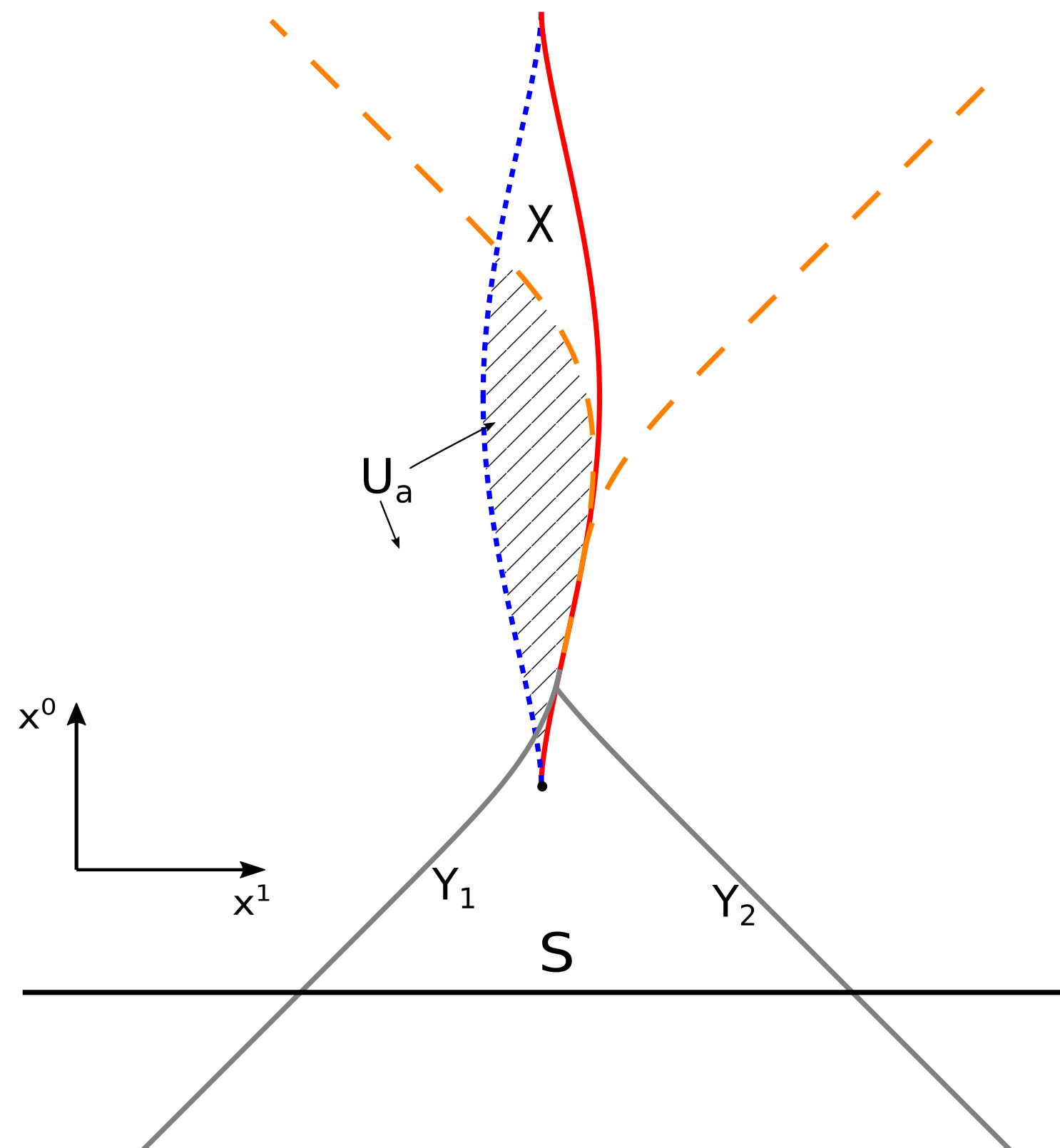}
		\caption{The region $U_a=\Psi(V_a)$ contains part of the left boundary of $X$ and extends up to the right boundary of $X$, where the gradient of the solution $\Phi_a$ diverges. The future (w.r.t. $g$) boundary of $U_a$ consists of the large dashed orange curves (null w.r.t. $g$) and a section of the right boundary of $X$ (spacelike w.r.t. $g$) starting at the black dot.}
		\label{fig:Exieta}
	\end{minipage}
\end{figure}

{\it Step 4.} Finally we show that there is a {\it different} way of extending $\Phi: U \rightarrow \mathbb{R}$ to give a GHD and that this implies non-uniqueness of MGHDs.

We construct this new extension of $\Phi: U \rightarrow \mathbb{R}$ as follows. Define $V_b$ to be the reflection of $V_a$ under $y^1 \rightarrow -y^1$. So $V_b$ is an extension of $V$ into region $F$. Everything we've said about $V_a$ is true also of $V_b$ and so this defines another GHD $\Phi_b : U_b \rightarrow \mathbb{R}$ where $U_b = \Psi(V_b)$. In this case, $U_b$ extends across the {\it right} boundary of $X$ all the way the the left boundary of $X$, where the gradient of $\Phi_b$ diverges when approaching from the right.\footnote{In the Nambu-Goto string interpretation, this corresponds to Fig. \ref{fig:ngstringprofilea} with point $B$ on the left boundary of $X$.}

We now have two different GHDs of the same intial data on $S$, $\Phi_a: U_a \rightarrow \mathbb{R}$ and $\Phi_b : U_b \rightarrow \mathbb{R}$. These two solutions agree in $U$ but they differ in $X$ because $\Phi_a$ has divergent gradient on the right boundary of $X$ whereas $\Phi_b$ has divergent gradient on the left boundary of $X$. Thus the corresponding maximal GHDs must differ in $X$. {\it This demonstrates the non-uniqueness of maximal GHDs for \eqref{Eqc1}.} 

We will now discuss this result and highlight properties of our example that are relevant to the general results of Section \ref{theorems}.

Consider the intersection $U_a \cap U_b$ shown in Fig. \ref{fig:intersection}. Note that this is {\it disconnected}, consisting of two connected components. One component contains $S$ but no points of $X$ and the other component is a subset of $X$. The two solutions agree on the former component but they disagree on the latter component. In Section \ref{theorems} we will prove that this disconnectedness is a necessary condition for two GHDs to differ in some region.  

\begin{figure}
	\centering
	\includegraphics[width=0.4\linewidth]{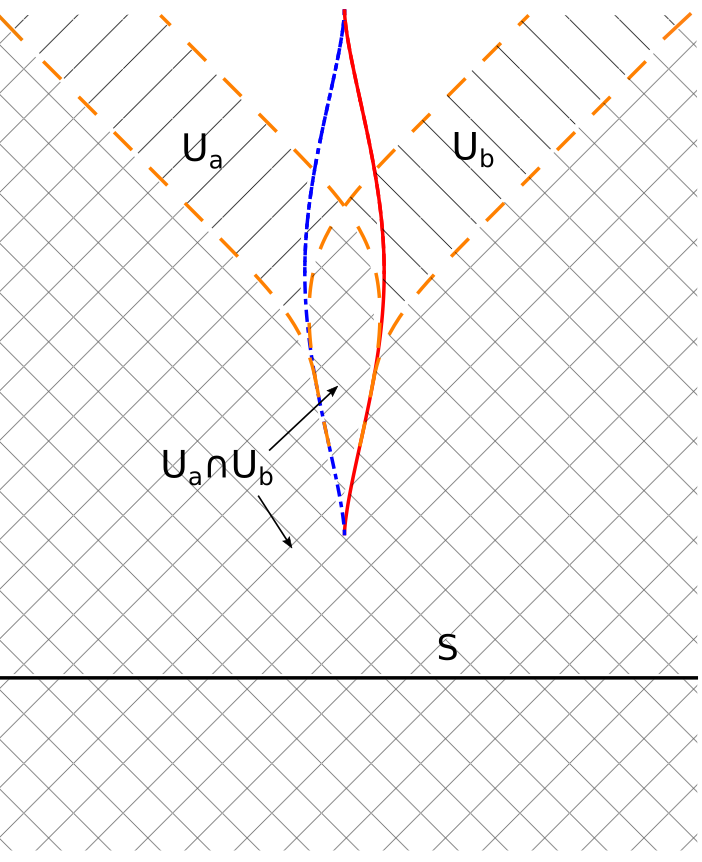}
	\caption{The regions $U_a$ and $U_b$ are given by the right/left hatching respectively. The intersection of these regions is disconnected, with one component lying inside $X$ and the other component (containing $S$) outside $X$.}
	\label{fig:intersection}
\end{figure}

Another point to emphasize is that the boundary of $U_a$ consists of a section (along the right boundary of $X$, between the lower black dot and the orange curves of Fig. \ref{fig:Exieta}), which can be approached from both sides (either the left or the right) within $U_a$. In other words $U_a$ {\it lies on both sides of its boundary}. (The same is true for $U_b$.) In section \ref{theorems} we will show that this property is a necessary condition for non-uniqueness of MGHDs. 

We have shown that there exist two distinct MGHDs arising from the same data on $S$. In fact one can show that there are infinitely many such MGHDs (cf section \ref{sec:nonuniqueness}). The different MGHDs all agree in the region $U$ but they differ in $X$. In section \ref{theorems} we will define the {\it maximal unique globally hyperbolic development} (MUGHD) $\Phi_{\rm max}:R \rightarrow \mathbb{R}$ of the initial data on $S$ as follows. $R$ is the largest open subset of Minkowski spacetime on which the solution is {\it uniquely} determined by the data on $S$. Such a development is necessarily globally hyperbolic with Cauchy surface $S$. For the above example, we have $R = U$ and $\Phi_{\rm max} = \Phi$. As we have seen, the solution $\Phi:U \rightarrow \mathbb{R}$ can be extended, whilst maintaining global hyperbolicity, but not in a unique way. From Figs \ref{fig:Urhosigma} and \ref{fig:Uxieta} we see that the future boundary of $R$ consists of a singular point (the lower black dot in Fig. \ref{fig:Uxieta}) from which emanate a pair of spacelike (w.r.t. $g$) curves which connect to a pair of null (w.r.t. $g$) curves. The solution can be smoothly, but not uniquely, extended across these spacelike and null curves.

The extendibility across the spacelike curves is a new kind of breakdown of predictability. Fig. \ref{fig:Uxieta} suggests that we should view these spacelike curves (the early time sections of the red and blue dotted curves) as a ``consequence" of the formation of a singularity (the black dot). This interpretation is suggested if one uses $x^0$ as a time function (e.g. in a numerical simulation). However, since these curves are spacelike, they are not in causal contact with the singularity. Furthermore, it is just as legitimate to use $y^0$ as a time function. From this point of view, Fig. \ref{fig:Urhosigma} shows that the spacelike curves form before (i.e. at earlier $y^0$) the singular point. So it is incorrect to ascribe the breakdown of predictability to the formation of the singularity. 

This behaviour is worrying. Given a development of the data on $S$, there is no general way of determining, {\it from the solution itself},  which region of it belongs to the MUGHD. To determine this region one has to construct all GHDs with the same initial data! This is much worse than the failure of predictability associated with the formation of a Cauchy horizon because the location of a Cauchy horizon within a development can be determined from the solution itself. 

How would the non-uniqueness of MGHDs manifest itself in, say, a numerical simulation? The answer is that the solution will depend not just on the initial data but also on the choice of time function. To see this, consider the globally hyperbolic development $\Phi_a: U_a \rightarrow \mathbb{R}$. Since $S$ is a Cauchy surface we can choose a global time function for $U_a$ such that $S$ is a surface of constant time. We can do the same for $\Phi_b : U_b \rightarrow \mathbb{R}$. Of course these two time functions are different but either could be used for a numerical evolution starting from the data on $S$. For points in the MUGHD $U$, the results of these two numerical evolutions will agree. However, for points in $X$, the results will disagree. In practice one would not know {\it a priori} which points belong to the MUGHD, i.e., one would not know in what region the results of the numerical evolution are independent of the choice of time function.\footnote{Since we are dealing with a subluminal theory, one could just declare that $x^0$ is a preferred time function and ignore the above problems. However this is unsatisfactory: if one uses $x^0$ as the time function (with $S$ a surface of constant $x^0$ at sufficiently early time) then from Fig. \ref{fig:Uxieta} the evolution must stop at the line of constant $x^0$ passing through the singularity corresponding to the (lower) black dot so one obtains only part of the MUGHD.} 

Note that, for any solution, the domain of dependence of $S$ defined using the Minkowski metric $m$ is a subset of the domain of dependence of $S$ defined using $g$. Hence a solution which is globally hyperbolic w.r.t. $g$ is also globally hyperbolic w.r.t. $m$. We could therefore ask about uniqueness of MGHDs defined w.r.t. $m$ instead of w.r.t. $g$. We'll refer to these as $m$-MGHDs. For the above example, there is indeed a unique $m$-MGHD: it is bounded to the future by two future-directed null (w.r.t. $m$) lines emanating from the lower black dot in Fig. \ref{fig:Uxieta}. We'll prove in Section \ref{theorems} that any subluminal equation always admits a unique $m$-MGHD, which is a subset of the MUGHD. However, if the speed of propagation w.r.t. $g$ is much less than the speed of propagation w.r.t. $m$ then the $m$-MGHD will not be a very useful concept because it will not contain a large part of the MUGHD. 

We have used the Born-Infeld scalar as an example exhibiting non-uniqueness of MGHDs. This example is rather artificial because there is a ``more fundamental" underlying theory, namely the Nambu-Goto string, for which there is no problem with predictability. However, our point is that if this pathological feature can occur for a particular scalar field theory then it is to be expected to occur also for other scalar field theories for which there is no analogue of the Nambu-Goto string interpretation. 

This ends the heuristic discussion of our example of non-uniqueness. We will now present a rigorous proof of the non-uniqueness of MGHDs.\footnote{Note that the regions $V_a$, $V_b$ etc in the proof of this theorem are defined slightly differently from the regions defined in the discusion above.}

\subsubsection{Theorem on non-uniqueness of MGHDs}

\label{nutheorem}

\begin{theorem} \label{ThmNonUniqueness}
For the equation \eqref{Eqc1} there exist two GHDs $\Phi_a : U_a \to \R$ and $\Phi_b : U_b \to \R$  of the same initial data posed on $\{x^0=0\}$ such that there exists an $x \in U_a \cap U_b$ with $\Phi_a(x) \neq \Phi_b(x)$.
\end{theorem}

\begin{proof}
We begin by remarking that we will prove the statement of the theorem with the hypersurface $\{x^0=0\}$ replaced by $\{x^0 = -T\}$ for $T\gg 1$. This represents no loss of generality since the equation \eqref{Eqc1} is invariant under translations in $x^0$. We will construct the two GHDs using Theorem \ref{diffeo_theorem} and Lemma \ref{glob_hyp_lemma}. 

We choose $\Phi_1$ and $\Phi_2$ as in \eqref{phi_def} and recall equation \eqref{PhiPrime}. We start by investigating the map $\Psi(\rho, \sigma) = \big(\xi(\rho, \sigma), \eta(\rho, \sigma)\big)$ defined by (\ref{xisol}) and (\ref{etasol}).  
\vspace*{3mm}

\textbf{Step 1:}   \emph{Analysis of the level sets of $\eta(\rho, \sigma)$.}

We begin by noticing that the function 
\begin{equation} \label{Feta}
\eta(\rho, \sigma) = \sigma + \int_\rho^\infty(\phi'(x))^2\,dx
\end{equation}
clearly satisfies $d\eta (\rho, \sigma) = d\sigma - (\phi'(\rho))^2 d\rho \neq 0$ for all $(\rho, \sigma) \in \R^2$, and thus its level sets are closed, embedded, one-dimensional submanifolds which foliate $\R^2$. The leaf $\{\eta = \eta_0\}$ can be written as a graph over the $\rho$-axis: $\sigma_{\eta_0}(\rho) = \eta_0 - \int_\rho^\infty(\phi'(x))^2\,dx$. It follows from $\frac{d\sigma_{\eta_0}}{d\rho} = (\phi'(\rho))^2 >0$, that the graph is strictly monotonically increasing. Moreover, we have $\sigma_{\eta_0}(\rho) \to \eta_0 -C_0 $ for $\rho \to -\infty$ and $\sigma_{\eta_0}(\rho) \to \eta_0 $ for $\rho \to +\infty$, where $C_0 := \int_{-\infty}^{\infty} (\phi'(x))^2\,dx >0$.

Next we investigate the qualitative behaviour of the intersection of the leaves of constant $\eta = \eta_0$ with the circle of radius $r_0 = \sqrt{2 \ln (a)}$. Let $\eta|_{r=r_0}$ denote the restriction of $\eta$ to the circle of radius $r_0$. Since the latter is compact, it follows that $\eta|_{r=r_0}$ takes on its minimum $\eta_{\min} := \min \eta|_{r = r_0}$ and its maximum $\eta_{\max} := \max \eta|_{r = r_0}$. Hence, the differential of $\eta|_{r =r_0}$ must have at least two zeros.

We now parametrise the circle of radius $r_0$ by $\gamma_\pm(\rho) = (\rho, \pm \sqrt{2\ln(a) - \rho^2})$ where $\rho \in (-\sqrt{2 \ln(a)}, \sqrt{2\ln(a)})$, and we compute $\dot{\gamma}_\pm(\rho) = \partial_\rho \mp \frac{\rho}{\sqrt{2\ln(a) - \rho^2}} \partial_\sigma$. It follows that
\begin{equation}\label{DerivativeEta}
d\eta|_{r=r_0}\big(\dot{\gamma}_\pm(\rho)\big) = \mp \frac{\rho}{\sqrt{2 \ln(a) - \rho^2}} - \big(\phi'(\rho)\big)^2
\end{equation}
Let us first consider the upper arc of the circle, i.e., $\sigma >0$ and the minus sign in \eqref{DerivativeEta}. We try to solve
\begin{equation}\label{DerEta2}
\underbrace{-\frac{\rho}{\sqrt{2\ln(a) - \rho^2}}}_{=:f(\rho)} = a^2 e^{-2\rho^2}
\end{equation}
Clearly, this does not have any solution for $\rho \geq 0$. We also remark that regularising $\dot{\gamma}(\rho)$ for $\rho \to \sqrt{2\ln(a)}$ by multiplying it, and thus also \eqref{DerEta2}, by $\sqrt{2\ln(a) - \rho^2}$, shows that $(\rho = \sqrt{2\ln(a)}, \sigma = 0)$ cannot be an extremum of $\eta|_{r = r_0}$. It follows that $\eta|_{r = r_0}$ does not have any extrema in the quadrant $\{\rho \geq 0, \sigma \geq 0\}$.

On the other hand we have
\begin{equation*}
\frac{d}{d\rho} f(\rho) = -\frac{1}{\sqrt{2\ln(a) - \rho^2}} - \frac{\rho^2}{(2\ln(a) - \rho^2)^{\frac{3}{2}}} <0
\end{equation*}
and
\begin{equation*}
\frac{d}{d\rho} (a^2 e^{-2\rho^2}) = -2a^2 \rho e^{-2\rho^2} >0 \qquad \textnormal{ for } \rho <0 \;.
\end{equation*}
It follows that $\eta|_{r = r_0}$ can have at most one extremum in the quadrant $\{\rho < 0, \sigma > 0\}$.

Similarly, by considering the lower arc of the circle and the plus sign in \eqref{DerivativeEta} we find that $\eta|_{r=r_0}$ does not have any extrema in the quadrant $\{\rho \leq 0, \sigma \leq 0\}$ and at most one extremum in the quadrant $\{\rho >0, \sigma <0\}$. It thus follows that $\eta|_{r=r_0}$ has \emph{exactly} two extrema, one in the quadrant $\{\rho < 0, \sigma > 0\}$ and one in the quadrant $\{\rho >0, \sigma <0\}$. Moreover, it is easy to see that the extremum in the quadrant $\{\rho < 0, \sigma > 0\}$ is the global maximum (for example this follows from $d\eta|_{r=r_0}\big(\dot{\gamma}_+(0)\big) <0$) and the extremum in the quadrant  $\{\rho >0, \sigma <0\}$ is the global minimum. Furthermore it  follows that $\eta|_{r=r_0}$ is strictly monotonically increasing along the two segments of the circle of radius $r_0$ that connect the minimum of $\eta|_{r=r_0}$ with its maximum. In particular, $\eta|_{r=r_0}$ takes on every value $\eta \in (\eta_{\min}, \eta_{\max})$ exactly twice.

In conclusion, we have established the following qualitative picture:
The curves of constant $\eta$ for $\eta < \eta_{\min}$ are disjoint of the circle of radius $r=r_0$ and lie below it, the curve $\eta = \eta_{\min}$ touches the circle in exactly one point in the quadrant $\{\rho >0, \sigma <0\}$, the curves of constant $\eta$ for $\eta_{\min} < \eta < \eta_{\max}$ intersect the circle in exactly two points, the curve $\eta = \eta_{\max}$ touches the circle in exactly one point in the quadrant $\{\rho <0, \sigma >0\}$, and finally the curves of constant $\eta$ for $\eta > \eta_{\max}$ are disjoint of the circle of radius $r=r_0$ and are lying above it. This behaviour is summarised in Figure \ref{FigEtaCircle}.
\begin{figure}[h]
  \centering
  \includegraphics[width=8cm]{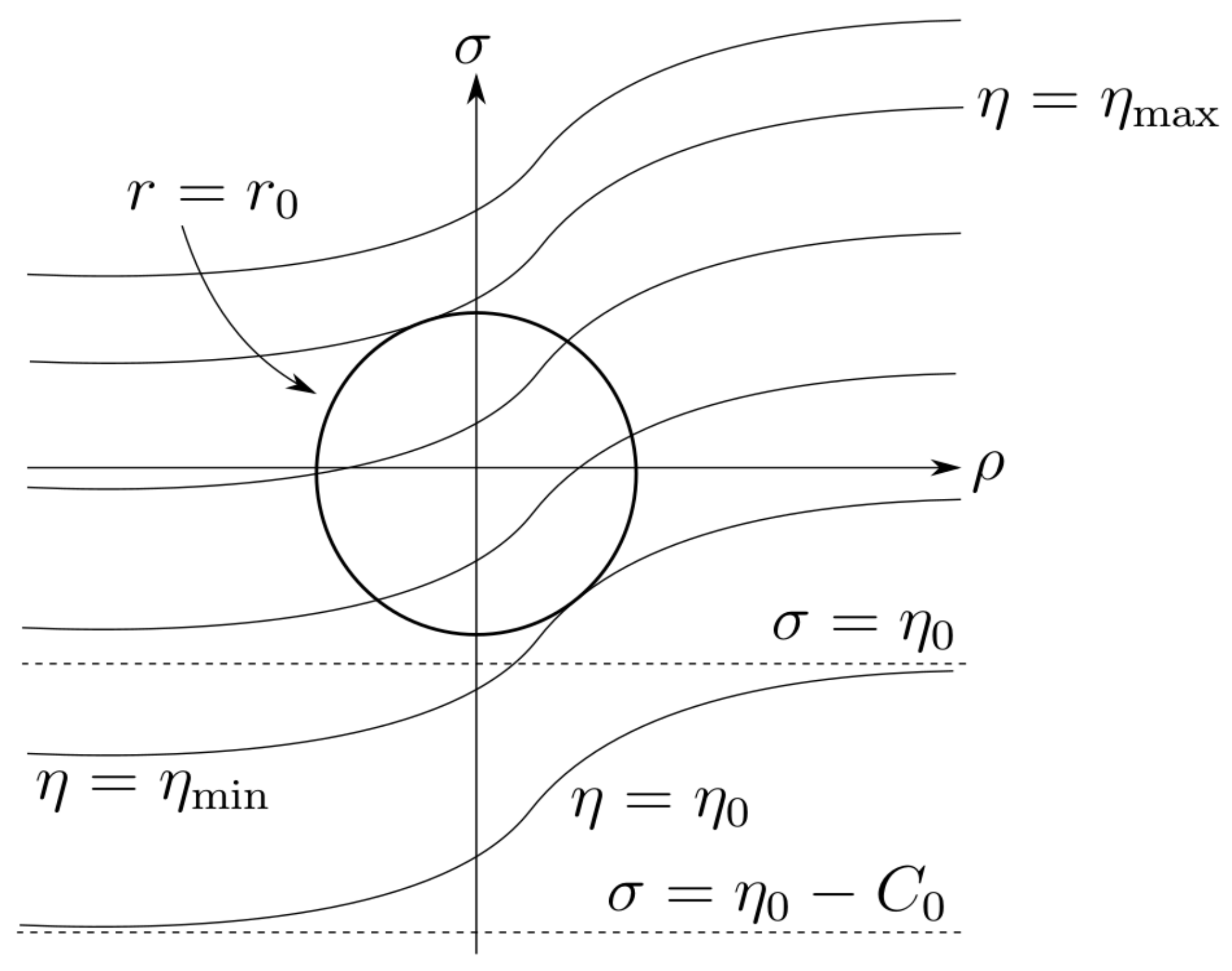}
       \caption{The level sets of $\eta$ in the $(\rho, \sigma)$ plane. Note that Lemmas \ref{glemma} and \ref{lemma_orientation} imply that rotating this diagram clockwise through $45^\circ$ gives a Penrose diagram on which the null geodesics of $g$ are straight lines at $45^\circ$ to the horizontal, with time increasing up the diagram.}
      \label{FigEtaCircle}
\end{figure}
\vspace*{3mm}

If $\Psi : \R^2_{\rho, \sigma} \supseteq V \to U \subseteq \R^2_{\xi, \eta}$ is a diffeomorphism, then Theorem \ref{diffeo_theorem} gives a solution $\Phi : U \to \R$ of \eqref{Eqc1}. By \eqref{jacobian} and \eqref{PhiPrime} $\Psi$ is a local diffeomorphism everywhere away from the circle $r = r_0$.
\vspace*{3mm}

\textbf{Step 2:}   \emph{We analyse  in which  regions  in $\R^2_{\rho, \sigma} $ the map  $\Psi$ is injective.}
 
For this, let us assume that for $(\hat{\rho}, \hat{\sigma}), (\tilde{\rho}, \tilde{\sigma}) \in \R^2_{\rho, \sigma}$ we have 
\begin{equation}\label{PsiIn}
\Psi(\hat{\rho}, \hat{\sigma})= \Psi(\tilde{\rho}, \tilde{\sigma}) =(\hat{\xi}, \hat{\eta})\;.
\end{equation} 
Obviously this entails that $(\hat{\rho}, \hat{\sigma})$ and $ (\tilde{\rho}, \tilde{\sigma})$ both have to lie on the same curve $\eta = \hat{\eta}$. Recall from \eqref{rhoeq} that along the curve  $\eta = \hat{\eta}$, the value of the coordinate $\xi$ is a function of $\rho$: $\xi = F(\rho; \hat{\eta})$. Also recall from \eqref{dFdrho} that we have
\begin{equation}\label{GradF}
\frac{\partial F}{\partial \rho}(\rho; \hat{\eta}) = 1 - \phi'(\rho)^2 \phi'(\sigma_{\hat{\eta}}(\rho))^2 \;.
\end{equation}
It now follows from our qualitative understanding of the curves of constant $\eta$ together with \eqref{PhiPrime}  that for $\hat{\eta} < \eta_{\min}$ or $\hat{\eta} > \eta_{\max}$ we have $\frac{\partial F}{\partial \rho} (\rho; \hat{\eta}) >0$ for all $\rho \in \R$. Thus, for such $\hat{\eta}$, $\xi$ is a strictly monotonically increasing function in $\rho$ along $\eta = \hat{\eta}$ and thus \eqref{PsiIn} implies $(\hat{\rho}, \hat{\sigma}) = (\tilde{\rho}, \tilde{\sigma})$. It is easy to see that $\xi$ is still strictly monotonically increasing along $\eta = \eta_{\min}, \eta_{\max}$, since those curves touch $r = r_0$ only in one point and thus the right hand side of \eqref{GradF} only vanishes at one point. Thus, we have shown 
\begin{equation}\label{Inj1}
\textrm{$\Psi$ is injective in the regions $\{\eta \leq \eta_{\min}\}$ and $\{\eta \geq \eta_{\max}\} \subseteq \R^2_{\rho, \sigma}$.}
\end{equation}

However, for $\hat{\eta} \in (\eta_{\min}, \eta_{\max})$ the right hand side of \eqref{GradF} is negative inside the circle of radius $r_0$. Let $-r_0 \leq \rho_{\mathrm{enter}}$ be the value of $\rho$ at which $\eta = \hat{\eta}$ enters the circle of radius $r_0$ and $\rho_{\mathrm{leave}} \leq r_0$ the value at which $\eta = \hat{\eta}$ leaves the circle of radius $r_0$.  Thus, the function $ \rho \mapsto F(\rho; \hat{\eta})$ is strictly monotonically decreasing for $\rho \in (\rho_{\mathrm{enter}}, \rho_{\mathrm{leave}})$. For $\rho \in (-\infty, \rho_{\mathrm{enter}}) \cup (\rho_{\mathrm{leave}}, \infty)$ the right hand side of \eqref{GradF} is positive and thus $ \rho \mapsto F(\rho; \hat{\eta})$ is strictly monotonically increasing for such $\rho$. Moreover, it follows from $\phi'(\rho)^2\phi'(\sigma)^2 = a^2 e^{-(\rho^2 + \sigma^2)}$ (see \eqref{phi_def} and below) that for $\rho \to \pm \infty$, the right hand side of \eqref{GradF} tends to $1$. Thus we have $F(\rho; \hat{\eta}) \to -\infty$ for $\rho \to -\infty$ and $F(\rho; \hat{\eta}) \to +\infty$ for $\rho \to +\infty$. The qualitative behaviour of the function $ \rho \mapsto F(\rho; \hat{\eta})$, which we have just established, is depicted in Figure \ref{FigF}.
\begin{figure}[h]
  \centering
  \includegraphics[width=6cm]{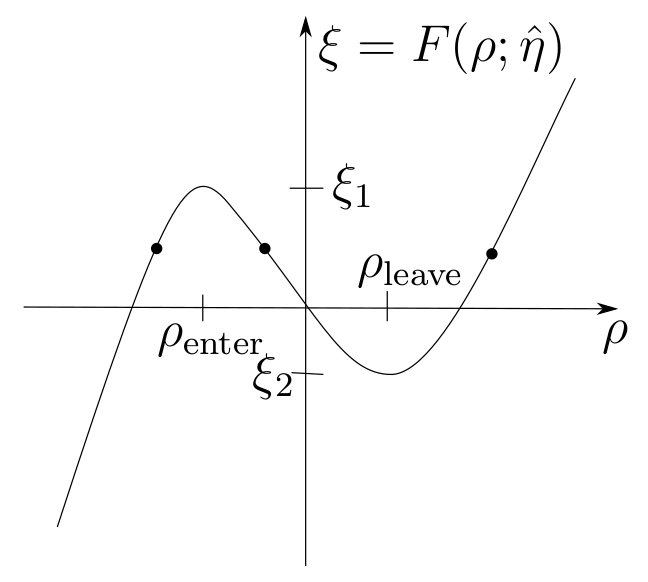}
      \caption{The qualitative behaviour of $ \rho \mapsto F(\rho; \hat{\eta})$ for $\hat{\eta} \in (\eta_{\min}, \eta_{\max})$. The black dots show that for $\hat{\xi} \in (\xi_1,\xi_2)$ there are three solutions $\rho_i$ of $\hat{\xi} = F(\rho,\hat{\eta})$. } \label{FigF}
\end{figure}
Hence, $\xi_1 := F(\rho_{\mathrm{enter}}; \hat{\eta})$ is a local maximum and $\xi_2 := F(\rho_{\mathrm{leave}}; \hat{\eta})$ is a local minimum of $F(\rho; \hat{\eta})$. For each $\hat{\xi} \in (\xi_1, \xi_2)$ there exist, by the intermediate value theorem and the monotonicity properties, exactly three points $\hat{\rho}_1< \rho_{\mathrm{enter}} < \hat{\rho}_2 <  \rho_{\mathrm{leave}} <\hat{\rho}_3$ with $\hat{\xi} = F(\hat{\rho}_i; \hat{\eta})$, $i = 1,2,3$. Setting $\hat{\sigma}_i = \sigma_{\hat{\eta}}(\hat{\rho}_i)$, this determines three distinct points $(\hat{\rho}_i, \hat{\sigma}_i)$, $i=1,2,3$, which are all mapped to $(\hat{\xi}, \hat{\eta})$ under $\Psi$. (Two of these points lie outside the circle $r=r_0$ and one lies inside.) Hence, we have shown that $\Psi$ is not injective in $\{\eta_{\min} < \eta < \eta_{\max}\}$. 

To understand the subsets of $\{\eta_{\min} < \eta < \eta_{\max}\}$ on which $\Psi$ \emph{is} injective, let us first observe that the set $\{\eta_{\min} < \eta < \eta_{\max}\} \setminus \{r \leq r_0\}$ has two connected components. We denote the `left' component\footnote{I.e., the component which contains points of arbitrarily large and negative $\rho$-coordinate.} by $\{\eta_{\min} < \eta < \eta_{\max}\}_L$ and the `right' component by $\{\eta_{\min} < \eta < \eta_{\max}\}_R$. It follows directly from \eqref{PhiPrime} and \eqref{GradF} that 
\begin{equation}\label{Inj2}
\textrm{$\Psi$ is injective on $\{\eta_{\min} < \eta < \eta_{\max}\}_L$ as well as on 
$\{\eta_{\min} < \eta < \eta_{\max}\}_R$.}
\end{equation}
Moreover, we note that inside the circle of radius $r_0$ the right hand side of \eqref{GradF} is bounded from below by $1 -a^4 <0$. Thus, for $\hat{\eta} \in (\eta_{\min}, \eta_{\max})$ the function $F(\rho; \hat{\eta})$ can decrease in between $\rho_{\mathrm{enter}}$ and $\rho_{\mathrm{leave}}$ by at most $2r_0 \cdot (1 - a^4)$. On the other hand, since $\phi'(\rho)^2\phi'(\sigma)^2 = a^2 e^{-r^2}$, the right hand side of \eqref{GradF} is bounded from below by a positive constant in $\{r > 2r_0\}$. It thus follows that we can choose $\rho_0 > 2r_0$  large enough such that $F(\rho_0; \hat{\eta}) - F(r_0; \hat{\eta}) > 2r_0 \cdot (1 - a^4)$ for all $\hat{\eta} \in (\eta_{\min}, \eta_{\max})$. Since we have $\rho_{\mathrm{leave}}(\hat{\eta}) \leq r_0$ for all $\hat{\eta} \in (\eta_{\min}, \eta_{\max})$ this shows
\begin{equation*}
\sup_{\rho < \rho_{\mathrm{enter}}}F(\rho; \hat{\eta}) < \inf_{\rho > \rho_0}F(\rho;\hat{\eta}) \;.
\end{equation*}
Hence, we have shown
\begin{equation}\label{Inj3}
\textrm{$\Psi$ is injective in $\{\eta_{\min} < \eta < \eta_{\max}\}_L \cup \Big(\{\eta_{\min} < \eta < \eta_{\max}\}_R \cap \{\rho > \rho_0\}\Big)$}.
\end{equation}
\vspace*{3mm}

\textbf{Step 3:}   \emph{Construction of the two solutions $\Phi_a : U_a \to \R$ and $\Phi_b : U_b \to \R$.}

Let now $(\rho_c, \sigma_c)$ denote the point of contact of $\eta = \eta_{\min}$ with the circle $r = r_0$. We have $0<\rho_c <r_0$ and $-r_0 < \sigma_c < 0$. We now define  the region 
\begin{equation}\label{Va}
V_a := \{\eta < \eta_{\min}\} \cup \{\sigma < -r_0\} \cup \Big( \{0<\rho < \rho_c\} \cap \{r > r_0\} \cap \{\sigma < \sigma_c\}\Big) \cup \{\rho > \rho_0\}\;,
\end{equation}
cf.\ Figure \ref{FigA}.
It follows from \eqref{Inj1} and \eqref{Inj3} that $\Psi$ is injective on $V_a$ and thus a diffeomorphism onto its image.
Recall that we have
\begin{equation*}
x^0 = \frac{1}{2}(\eta - \xi) = \frac{1}{2}\left[ \sigma - \rho + \int_\rho^\infty\big(\phi'(x)\big)^2 \, dx + \int_{-\infty}^\sigma \big(\phi'(x)\big)^2 \, dx \right] \;.
\end{equation*}
Thus the hypersurface $x^0 = -T$ is given in the $(\rho, \sigma)$ plane by
\begin{equation*}
-2T = \sigma - \rho + \int_\rho^\infty\big(\phi'(x)\big)^2 \, dx + \int_{-\infty}^\sigma \big(\phi'(x)\big)^2 \, dx \;.
\end{equation*}
It follows from the trivial bounds on the integrals that for $T \gg 1$ the hypersurface $x^0 = -T$ is contained in $\Psi(V_a)$.\footnote{The reason for including $\{\rho > \rho_0\}$ in \eqref{Va} was exactly to ensure this.}   We now define our first GHD $\Phi_a : U_a \to \R$ by setting $U_a = \Psi(V_a)$ and defining $\Phi_a$ by Theorem \ref{diffeo_theorem}. It is easy to convince oneself, using Lemma \ref{glob_hyp_lemma}, that the domain is indeed globally hyperbolic with Cauchy hypersurface $x^0 = -T$.

To define the second GHD $\Phi_b : U_b \to \R$, we set
\begin{equation*}
V_b := \{ \eta < \eta_{\min}\} \cup \{ \rho > r_0\} \cup \Big(\{0>\sigma > \sigma_c\} \cap \{r > r_0\} \cap \{\rho > \rho_c\}\Big) \;.
\end{equation*}
See also Figure \ref{FigB}. It follows from \eqref{Inj1} and \eqref{Inj2} that $\Psi$ is injective on $V_b$ and thus we can set $U_b = \Psi(V_b)$ and define $\Phi_b$ by Theorem \ref{diffeo_theorem}. Again, it is easy to convince oneself that the domain is indeed globally hyperbolic with Cauchy hypersurface $x^0 = -T$. 
\begin{figure}[h]
\centering
\begin{minipage}{.5\textwidth}
  \centering
  \includegraphics[width=6cm]{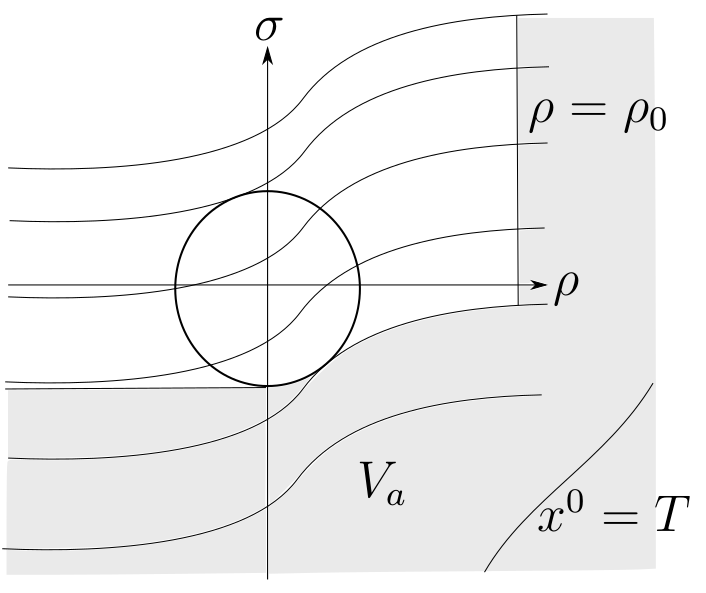}
      \caption{The domain $V_a$} \label{FigA}
\end{minipage}%
\begin{minipage}{.5\textwidth}
  \centering
  \includegraphics[width=6cm]{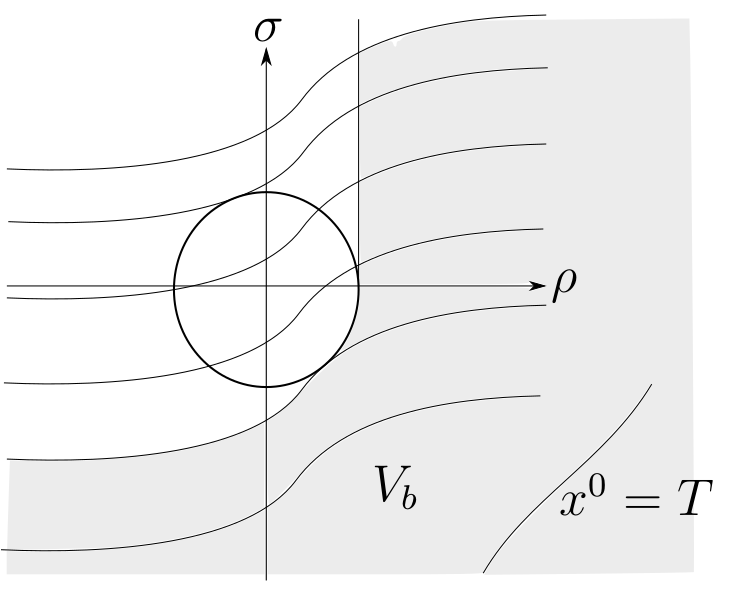}
	 \caption{The domain $V_b$} \label{FigB}
\end{minipage}
\end{figure}
\vspace*{3mm}

\textbf{Step 4:}   \emph{We show that there is an $x \in U_a \cap U_b$ with $\Phi_a(x) \neq \Phi_b(x)$.} 

For this consider a curve $\eta = \eta_{\min} + \varepsilon$. By continuity we can choose $\varepsilon>0$ small enough such that the curve $\sigma_{\eta_{\min} + \varepsilon}(\rho)$  intersects the circle of radius $r_0$ at $\rho_{\mathrm{enter}} < \rho_c < \rho_{\mathrm{leave}}$ and such that $\sigma_{\eta_{\min} + \varepsilon}(-\infty, \rho_{\mathrm{enter}} ) \subseteq V_a$ and $\sigma_{\eta_{\min} + \varepsilon}(\rho_{\mathrm{leave}}, \infty) \subseteq V_b$. Since $\eta_{\min} + \varepsilon \in (\eta_{\min}, \eta_{\max})$, the qualitative analysis of   $ \rho \mapsto  F(\rho;\eta_{\min} + \varepsilon) $ below \eqref{Inj1} applies which was summarised in Figure \ref{FigF}.
It follows that there exist $\rho_1 <\rho_{\mathrm{enter}} $ and $ \rho_{\mathrm{leave}}< \rho_2 < \rho_3$ and a strictly monotonically increasing function 
\begin{equation*}
\rho_{1\mathrm{enter}} : [0,1] \to [\rho_1, \rho_{\mathrm{enter}}]\textrm{ with }\rho_{1\mathrm{enter}}(0) = \rho_1\textrm{ and }\rho_{1\mathrm{enter}}(1) = \rho_{\mathrm{enter}}\;,
\end{equation*} and  a second strictly monotonically increasing function 
\begin{equation*}
\rho_{23} : [0,1] \to [\rho_{2}, \rho_3]\textrm{ with }\rho_{23}(0) = \rho_2\textrm{ and }\rho_{23}(1) = \rho_3
\end{equation*}
 such that   
 \begin{equation*}
 \Psi\big(\rho_{1\mathrm{enter}}(s), \sigma_{\eta_{\min} + \varepsilon}(\rho_{1\mathrm{enter}}(s))\big) = \Psi\big(\rho_{23}(s), \sigma_{\eta_{\min} + \varepsilon}(\rho_{23}(s))\big)
 \end{equation*}
 holds for all $s \in [0,1]$. 
 See also Figure \ref{FigF2} and the discussion following \eqref{Inj1}. 
\begin{figure}[h]
  \centering
  \includegraphics[width=6cm]{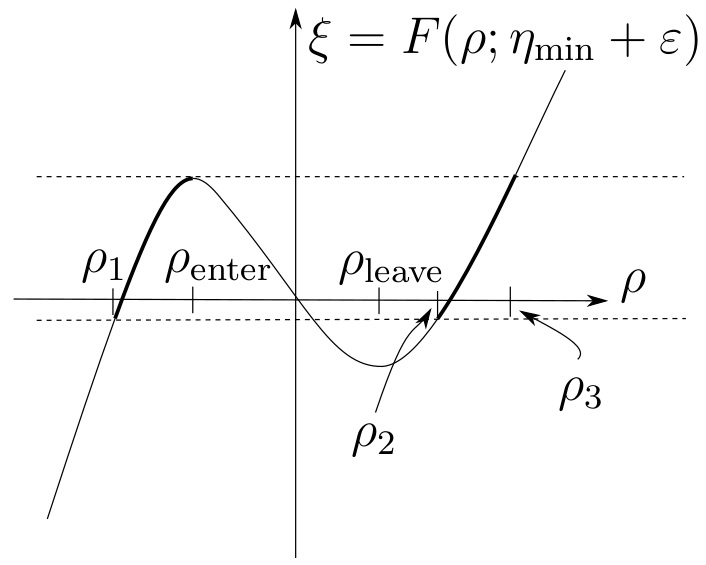}
       \caption{The functions $\rho_{1\mathrm{enter}}$ and $\rho_{23}$. } \label{FigF2}
\end{figure} 
 It now follows from Lemma \ref{lemma_sing} and \eqref{PhiPrime} that $\partial_\xi \Phi_a$ and $\partial_\eta \Phi_a $ along $\Psi\big(\rho_{1\mathrm{enter}}(s), \sigma_{\eta_{\min} + \varepsilon}(\rho_{1\mathrm{enter}}(s))\big)$ tend to $+\infty$ for $s \to 1$, while $\partial_\xi \Phi_b$ and $\partial_\eta \Phi_b $ along $\Psi\big(\rho_{23}(s), \sigma_{\eta_{\min} + \varepsilon}(\rho_{23}(s))\big)$ tend to a finite value. This suffices to establish the claim and thus conclude the proof.
\end{proof}

We remark the following:
\begin{remark}
We emphasise that the two GHDs constructed in the proof of Theorem \ref{ThmNonUniqueness} are \emph{smooth}. Hence, the non-uniqueness mechanism exhibited does not stem from a loss of regularity of the solution.
\end{remark}
\begin{remark}
We note that the two GHDs constructed in the proof of Theorem \ref{ThmNonUniqueness} are not maximal. However, an application of  Zorn's Lemma (!) to the set of globally hyperbolic extensions of $\Phi_a : U_a \to \R$ shows that there exists a MGHD which contains $\Phi_a : U_a \to \R$. In the same way one shows the existence of a MGHD which contains $\Phi_b : U_b \to \R$. These two MGHDs are clearly distinct. In fact one can even show that there are infinitely many distinct MGHDs of the constructed initial data. We leave the details to the reader. 
\end{remark}
\begin{remark}
The initial data constructed in the proof of Theorem \ref{ThmNonUniqueness} are not compactly supported. However, by a standard domain of dependence argument one can cut off the initial data outside a large enough ball to produce compactly supported initial data and two GHDs thereof which satisfy the statement of Theorem \ref{ThmNonUniqueness}.
\end{remark}
The following remark might be skipped and come back to when referred to later in Section \ref{theorems}.
\begin{remark} \label{RemCutOff}
Recall that equation \eqref{Eqc1} (which is \eqref{BIeom}  multiplied by $1+(\partial_{x^1}\Phi)^2 - (\partial_{x^0}\Phi)^2$) is not manifestly hyperbolic, i.e., a quasilinear wave equation of the form \eqref{QuasilinearEq} which we will consider in Section \ref{theorems}.   Its principal symbol is
\begin{equation*}
\tilde{g}^{-1} = (1+(\partial_{x^1}\Phi)^2 - (\partial_{x^0}\Phi)^2) \cdot g^{-1} = \begin{pmatrix}
-(1 + (\partial_{x^1}\Phi)^2) & \partial_{x^0}\Phi \partial_{x^1} \Phi \\ \partial_{x^0}\Phi \partial_{x^1} \Phi & 1 -(\partial_{x^0} \Phi)^2 
\end{pmatrix} 
\end{equation*}
and the determinant is $\det \tilde{g}^{-1} = -1 -(\partial_{x^1}\Phi)^2 + (\partial_{x^0}\Phi)^2 = -1 - 4\partial_{\eta}\Phi \partial_\xi \Phi$.
It follows from \eqref{NullDerivativesRhoSigma} and \eqref{phi_def} that for the two hyperbolic solutions constructed above we have
\begin{equation*}
\det \tilde{g}^{-1}  = -1 - 4 \frac{\phi'(\rho) \phi'(\sigma)}{(1-\phi'(\rho)\phi'(\sigma))^2} <-1 \;.
\end{equation*}
We can thus modify the principal symbol of \eqref{Eqc1} in the region $\{ (\partial_{x^0}\Phi, \partial_{x^1}\Phi)  \in \R^2\; |  \det \tilde{g}^{-1} \geq -\frac{1}{2} \} = \{(\partial_{x^0}\Phi)^2 \geq (\partial_{x^1} \Phi)^2 + \frac{1}{2} \}$ to make it Lorentz-metric valued for all $d\Phi$ (and keeping it subluminal), thus creating a subluminal quasilinear wave equation (i.e. a quasilinear equation for which every solution is hyperbolic) which is still solved by our two GHDs of the same initial data constructed above which take different values at a point that lies in both of their domains. 
\end{remark}

\subsection{Uniqueness for superluminal case}

Non-uniqueness of MGHDs is \textit{not} a problem in the superluminal ($c=-1$) case. We will prove this for an arbitrary superluminal equation in section \ref{theorems} below. In this section we will discuss briefly the interpretation of the example \eqref{phi_def} in the superluminal case.

In the superluminal case, recall that the Minkowski metric \eqref{minkbi} is
\be
 m = -(dx^1)^2 + (dx^0)^2
\ee
and we choose time orientation $\partial/\partial x^1$. Defining coordinates $(y^0,y^1)$ in the $(\rho,\sigma)$ plane as in \eqref{ydef} gives, for the flat metric of Lemma \ref{glob_hyp_lemma}
\be
\hat{m} = -(dy^1)^2 + (dy^0)^2
\ee
Now consider the example \eqref{phi_def}. We want to construct a solution using Theorem \ref{diffeo_theorem} so assume that $\Psi: V \rightarrow U$ is a diffeomorphism. Lemma \ref{nec_cond} implies that either $V \subset D$ or $V$ lies outside $D$. We consider the latter case, so $\Phi_1'(\rho)^2 \Phi_2'(\sigma)^2<1$ in $V$. The proof of Lemma \ref{lemma_orientation} reveals that $y^1$ is a global time function for $(U,g)$ (or $(V,\hat{m})$) so we take our initial surface $S=\Psi(\Sigma)$ where $\Sigma$ is a line $y^1 = -Y$ where $Y$ is large enough so that $\Sigma$ lies to the past of $\overline{D \cup E \cup F}$ as shown in Fig. \ref{fig:superrhosigma}.\footnote{Equivalently we could define $S$ to be a surface $x^0 = -T$ where $T$ is chosen large enough to make $S$ spacelike w.r.t. $g$. However, this would gives plots with a lot of white space between $S$ (or $\Sigma$) and the region of interest.} 

\begin{figure}
\centering
	\begin{minipage}{0.45\textwidth}
		\centering
		\includegraphics[width=0.95\linewidth]{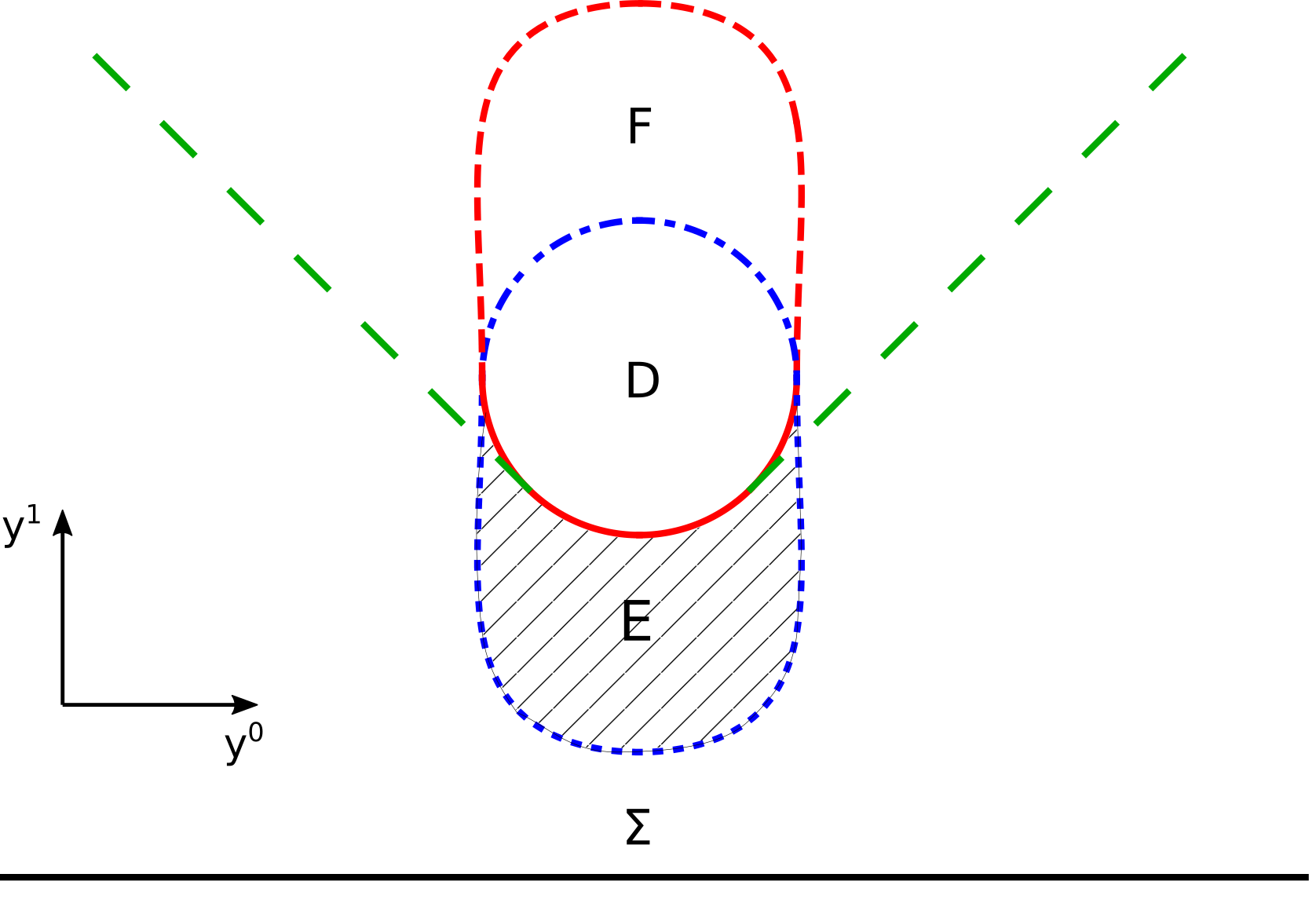}
		\caption{In the superluminal case we orient the plot so that the time function $y^1$ is the vertical axis. The large dashed green lines are lines of constant $\rho$ or $\sigma$ that are tangent to the circle at their point of contact. The MGHD of the data on $S$ is defined by choosing $V$ to be the region bounded to the future by the pair of large dashed green lines together with the section of the solid red curve joining them. This includes the hatched section of $E$ but not the two small regions of $E$ between the large dashed green lines and $D$.
		}
		\label{fig:superrhosigma}
		
	\end{minipage}\hfill
	\begin{minipage}{0.45\textwidth}
		\centering
		\includegraphics[width=0.95\linewidth]{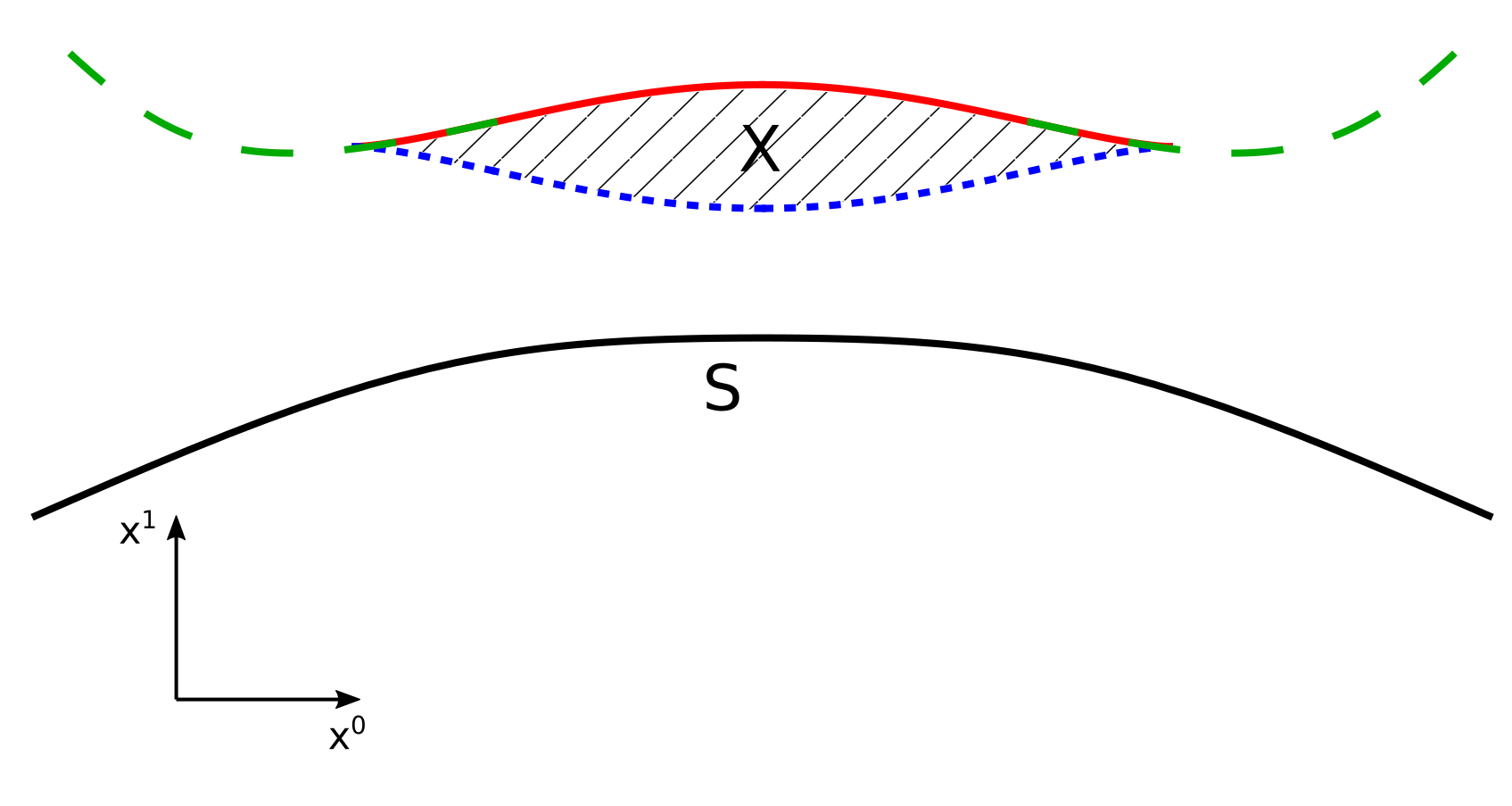}
		\caption{Plot of $U = \Psi(V)$ in Minkowski spacetime, oriented so that $x^1$ is the vertical axis. The MGHD is the region bounded to the future by the spacelike (w.r.t. $g$) solid red curve and the pair of null (w.r.t. $g$) large dashed green curves. This includes most of the region $X$. The gradient of $\Phi$ diverges on the solid red curve. The solution can be smoothly extended across the large dashed green curves, but not as a GHD of the data on $S$.
		 }
		\label{fig:superxieta}
	\end{minipage}
\end{figure}
The unique MGHD is obtained by taking $V$ to be the region defined in Fig. \ref{fig:superrhosigma}. The future boundary of $V$ is the union of a spacelike curve (a segment of the boundary of $D$) along which the gradient of $\Phi$ diverges (by Lemma \ref{lemma_sing}), and a pair of null curves across which the solution is smoothly extendible (but not as a GHD).\footnote{Note that the extendibility across the null sections of the boundary implies that the analogue of the strong cosmic censorship conjecture is false for the superluminal equation. But the behaviour is much better than in the subluminal case for which the object one needs to define to formulate this conjecture (the MGHD) is not even unique!} The corresponding picture in Minkowski spacetime is shown in Fig. \ref{fig:superxieta}. 

The reason that there is a unique MGHD in the superluminal case but not in the subluminal case was identified in Lemma \ref{sep_lemma}. In the subluminal case, different GHDs can be constructed by including points from $E$ or from $F$, or from both. But in the superluminal case, Lemma \ref{sep_lemma} implies that $F$ lies to the future of $D$ so from any point of $F$ there is a past directed timelike curve that ends on the boundary of $D$ and hence does not cross $\Sigma$. So no point of $F$ can belong to the domain of dependence of $\Sigma$.

\subsection{Higher dimensions}

It is easy to see that the pathological behaviour in the subluminal case is not restricted to two spacetime dimensions. The Born-Infeld scalar field theory in $(d+1)$-dimensional Minkowski spacetime is defined by generalizing the action (\ref{BIaction}) to $d+1$ dimensions. The two dimensional theory can be obtained trivially from the $d+1$ dimensional theory by assuming that $\Phi$ does not depend on $d-1$ of the spatial coordinates. Hence our 2d solutions can be interpreted as solutions in $d+1$ dimensions with translational invariance in $d-1$ directions. Such solutions do not decay at infinity. However, given initial data for such a solution, one could modify the data outside a ball of radius $R$ so that it becomes compactly supported. In the subluminal case, the resulting solution would be unchanged in the region inside the ingoing Minkowski lightcone emanating from the surface of this ball. Hence if $R$ is chosen large enough then the evolution of the solution inside the ball will behave as discussed above for long enough to see non-uniqueness of MGHDs. 

In the higher-dimensional superluminal case, there is a unique MGHD: we will prove below that {\it any} superluminal equation always admits a unique MGHD.

\section{Uniqueness properties of the initial value problem for quasilinear wave equations}
\label{theorems}

\subsection{Introduction}
\label{theorems_intro}

In this section we consider a quasilinear wave equation of the form
\begin{equation}\label{QuasilinearEq}
g^{\mu \nu}(u,  du) \partial_\mu \partial_\nu u = F(u, du) \;,
\end{equation}
where $u : \R^{d+1} \supseteq U \to \R$, $g$ is a smooth Lorentz metric valued function,\footnote{All the results presented in this section generalise literally unchanged to the setting of Section \ref{general}, where one does not assume that $g$ in \eqref{QuasilinearEq} is a Lorentz metric valued function, but one restricts consideration to \emph{hyperbolic solutions}, i.e., solutions of \eqref{QuasilinearEq} for which $g$ \emph{is} Lorentz metric valued. The only slight modification necessary is for the proof of the local existence result, Theorem \ref{ThmEx}. Here one can for example cut off the principal symbol of the quasilinear equation in the fashion of Remark \ref{RemCutOff}  to create a quasilinear wave equation and then apply the local existence result for quasilinear wave equations to show local existence of hyperbolic solutions for quasilinear equations with hyperbolic initial data.} $F$ is smooth with $F(0,0) =0$, and the coordinates used for defining \eqref{QuasilinearEq} are the canonical coordinates $x^\mu$ on $\R^{d+1}$. 

Let $S \subseteq \R^{d+1}$ be a connected hypersurface of $\R^{d+1}$. \emph{Initial data for \eqref{QuasilinearEq} on $S$} consists of a smooth real valued function $f_0 : S \to R$ and a smooth one form $\alpha_0$ (with values in $ T^*\R^{d+1}$) along $S$ such that  $X(f_0) = \alpha_0(X)$ holds for all vectors $X$ tangent to $S$ and such that the hypersurface $S$ is spacelike with respect to the Lorentzian metric $g(f_0, \alpha_0)$.  A \emph{globally hyperbolic development} (GHD) of initial data $(f_0, \alpha_0)$ on a hypersurface $S$ for \eqref{QuasilinearEq} consists of a smooth solution $u : U \to \R$ of \eqref{QuasilinearEq} ($U \subseteq \R^{d+1}$ being open) with $S \subseteq U$ and $u|_S = f_0$, $du|_S = \alpha_0$, and such that $U$ is globally hyperbolic with respect to the Lorentzian metric $g(u, d u)$ with Cauchy hypersurface $S$.

As we will show/recall in the following, the initial value problem for the equation \eqref{QuasilinearEq} with initial data given on a hypersurface $S$ is \emph{locally well-posed}. Here, we mean by this that the following two properties hold:
\begin{enumerate}
\item  there exists a globally hyperbolic development $u : U \to \R$ of the initial data
\item given two globally hyperbolic developments $u_1 : U_1 \to \R$ and $u_2 : U_2 \to \R$ of the same initial data, then there exists a \emph{common globally hyperbolic development} (CGHD), that is, a globally hyperbolic development $v : V \to \R$ of the initial data with $V \subseteq U_1 \cap U_2$ and $v = u_1|_V = u_2|_V$.
\end{enumerate} 
Note that the second property is only a weak version of  what one might understand under `local uniqueness', since it allows for the existence of a third globally hyperbolic development $u_3 : U_3 \to \R$ of the same initial data such that there exists an $x \in V \cap U_3$ with $u_3(x) \neq u_1(x) = u_2(x)$.\footnote{We will discuss how this might happen at the end of Section \ref{SecMUGHD}.}

The aim of this section of the paper is to investigate the uniqueness properties for solutions of quasilinear wave equations. In Section \ref{SecUniqueness} we first prove the second property of the local well-posedness statement from above and then establish the main theorem of this section: two globally hyperbolic developments of the same initial data agree on the intersection of their domains \emph{if this intersection is connected}. Section \ref{SecSupLum} then specialises to quasilinear wave equations \eqref{QuasilinearEq} with the property that 
\begin{equation}
\label{Property}
\parbox{0.75\textwidth}{there exists a vector field  $T$ on $\R^{d+1}$ such that $T$ is timelike with respect to  $g^{\mu \nu}(u, d u)$ for all $u, du$. }
\end{equation} 
In particular superluminal equations have this property. We show that for such equations the intersection of the domains of two globally hyperbolic developments of the same initial data is always connected -- and we thus obtain that any two globally hyperbolic developments agree on the intersection of their domains. The case of subluminal equations is considered in Section \ref{SecUniquenessSub}. Here, we show that if one of the two globally hyperbolic developments is also \emph{globally hyperbolic with respect to the Minkowski metric}, then again, the intersection of the domains is connected -- and we can thus apply our main theorem from Section \ref{SecUniqueness}. 

The next three sections deal with existence questions: Section \ref{SecEx} proves the first property of the above local well-posedness statement, Section \ref{SecUMGHD} establishes the existence of a unique maximal globally hyperbolic development for quasilinear wave equations with the property \eqref{Property}, and Section \ref{SecMUGHD} considers subluminal equations and shows the existence of a maximal region on which solutions are unique and which is globally hyperbolic (i.e. a MUGHD). 

The final section, Section \ref{SecFin}, present a uniqueness criterion for \emph{general} quasilinear wave equations of a very different flavour. It states that if there exists a maximal globally hyperbolic development with the property that its domain of definition always lies to just one side of its boundary, then this maximal globally hyperbolic development is the unique one. In particular this implies uniqueness of the MGHD constructed in Ref. \cite{Chr}. 

\subsection{Uniqueness results for general quasilinear wave equations} \label{SecUniqueness}

\begin{proposition}[Local uniqueness]\label{PropLocUniqueness}
Let $ u_1 : U_1 \to \R$ and $u_2 : U_2 \to \R$ be two globally hyperbolic developments for \eqref{QuasilinearEq} of the same initial data prescribed on a hypersurface $S \subseteq \R^{d+1}$. Then there exists a common globally hyperbolic development $v : V \to \R$.
\end{proposition}

\begin{proof}
For $p \in S$ let $W_p \subseteq \R^{d+1}$ be an open neighbourhood of $p$ on which there exists slice coordinates for $S$ and in which the Lorentzian metric $g(f_0, \alpha_0)$ given by the initial data is $C^0$-close to the Minkowski metric. Moreover, we require $W_p \subseteq U_1 \cap U_2$. Let $S_p$ be an open neighbourhood of $p$ in $S$ the closure of which is compactly contained in $W_p$. 
The standard literature methods (see for example \cite{Sogge}) ensure that there is an open  neighbourhood $DS_p \subseteq W_p$ of $S_p$ with the property that any two solutions, which are defined on $DS_p$ and attain the given initial data on $S_p$, agree, and such that $DS_p$ is globally hyperbolic with Cauchy hypersurface $S_p$. It thus follows that $u_1|_{DS_p} = u_2|_{DS_p}$. We now set $V = \bigcup_{p \in S} DS_p$. It is immediate that $u_1$ and $u_2$ agree on this set and that $V$ is globally hyperbolic with Cauchy hypersurface $S$.
\end{proof}

One can now ask whether \emph{global uniqueness} holds, which is the property that if $u_1 : U_1 \to \R$ and $u_2 : U_2 \to \R$ are two globally hyperbolic developments of the same initial data, then $u_1$ and $u_2$ agree on $U_1 \cap U_2$. Note that `global' refers to the property that `the two solutions agree in \emph{all of $U_1 \cap U_2$}' -- in contrast to the local result provided by 
Proposition \ref{PropLocUniqueness}, which only guarantees uniqueness in some smaller subset of $U_1 \cap U_2$.

The last author sketched an idea for a proof of global uniqueness in Section 1.4.1 of \cite{sbierski}. However, this sketch has the flaw that it tacitly assumes that given two globally hyperbolic developments $u_1 : U_1 \to \R$ and $u_2 : U_2 \to \R$ of the same initial data, that $U_1 \cap U_2$ is then connected -- which is in general not true as illustrated by the example presented in Section \ref{SecIVP} of this paper, see in particular Remark \ref{RemCutOff}. The necessity of the assumption of connectedness enters in the sketch as follows: One starts by considering the maximal globally hyperbolic region $W$ contained in $U_1 \cap U_2$ on which $u_1$ and $u_2$ agree (i.e. the \emph{maximal common globally hyperbolic development} (MCGHD)) and one would like to show that this region coincides with $U_1 \cap U_2$. Assuming $W \subsetneq U_1 \cap U_2$ one can find a boundary point of $W$ in $U_1 \cap U_2$ \emph{provided $U_1 \cap U_2$ is connected}. The argument then proceeds by constructing a spacelike slice through a suitable boundary point and  appealing to the local uniqueness result in order to conclude that $u_1$ and $u_2$ also agree on a neighbourhood of this slice and thus on an even bigger globally hyperbolic region than $W$ -- a contradiction to the maximality of $W$. This is roughly how one proves global uniqueness \emph{under the condition that $U_1 \cap U_2$ is connected}. Note that if $U_1 \cap U_2$ is disconnected, the same argument shows that the domain $W$ of the MCGHD equals the connected component of $U_1 \cap U_2$ that contains $S$.

For the Einstein equations one does not need to condition the global uniqueness statement, since one has the freedom to construct the underlying manifold -- there is no fixed background. We will explain this in the following: Given two globally hyperbolic developments $u_1$ and $u_2$ for the Einstein equations one constructs a bigger one in which both are contained (and thus proves global uniqueness) by glueing $u_1$ and $u_2$ together along the MCGHD of $u_1$ and $u_2$. However, in the case that $u_1$ and $u_2$ are two globally hyperbolic developments of a quasilinear wave equation on a fixed background such that $U_1 \cap U_2$ is disconnected, glueing them together along the MCGHD (which equals the connected component of $U_1 \cap U_2$ which contains the initial data hypersurface), would yield a solution which is no longer defined on a subset of $\R^{d+1}$, but instead on a manifold which projects down on $U_1 \cup U_2 \subseteq \R^{d+1}$ and contains the other connected components of $U_1 \cap U_2$ twice. Of course this is not allowed if we insist that solutions of \eqref{QuasilinearEq} should be defined on a subset of $\R^{d+1}$. So the key difference between the Einstein equations and a quasilinear wave equation \eqref{QuasilinearEq} is that for the former the underlying manifold is constructed along with the solution whereas for the latter, it is fixed {\it a priori}. This is the reason why one does not need to condition the global uniqueness statement for the Einstein equations.


\begin{theorem}\label{ThmUnique}
Let $u_1 : U_1 \to \R$ and $u_2 : U_2 \to \R$ be two globally hyperbolic developments of \eqref{QuasilinearEq} arising from the same initial data given on a connected hypersurface $S \subseteq \R^{d+1}$. Assume that $U_1 \cap U_2$ is connected. Then $u_1$ and $u_2$ agree on $U_1 \cap U_2$.
\end{theorem}

The proof is based on ideas found in \cite{cbgeroch}, \cite{Ring13}, and \cite{sbierski}. 

\begin{proof}
\textbf{Step 1:} \emph{We construct the maximal common globally hyperbolic development of $u_1 : U_1 \to \R$ and $u_2 : U_2 \to \R$ .}

Given two globally hyperbolic developments $u_1 : U_1 \to \R$ and $u_2 : U_2 \to \R$ of \eqref{QuasilinearEq} arising from the same initial data on $S$, we consider the set $\{v_\alpha : V_\alpha \to \R \; | \; \alpha \in A\}$ of all common globally hyperbolic developments. By Proposition \ref{PropLocUniqueness} we know that this set is non-empty. We define $v_0 : V_0 \to \R$, where $V_0 := \bigcup_{\alpha \in A} V_\alpha$ and $v_0 (x) = v_\alpha (x)$ for $x \in V_\alpha$. It is immediate that this is well-defined and that $v_0 : V_0 \to \R$ is a common globally hyperbolic development with the property that any other common globally hyperbolic development is a subset of $V_0$. We call $v_0 : V_0 \to \R$ the \emph{maximal common globally hyperbolic development}.  

We now set out to show that $V_0 = U_1 \cap U_2$, from which the theorem follows. Assume that $V_0 \subsetneq U_1 \cap U_2$. Since we assume that $U_1 \cap U_2$ is connected, there exists then a point $q \in \partial V_0 \cap U_1 \cap U_2$.\footnote{$\partial V_0$ denotes the boundary of $V_0$ in $\R^{d+1}$.} Without loss of generality we assume that $q \in J^+_{g(u_1, du_1)}(S, U_1)$.\footnote{The notation $J^+_{g(u_1, du_1)}(S, U_1)$ denotes the causal future of $S$ in $U_1$ with respect to the Lorentzian metric $g(u_1, du_1)$. The notations for the causal past $J^-$, and timelike future/past $I^{\pm}$ are analogous. We refer the reader to \cite{ONeill} for the basic notions of causality theory.}\vspace*{3mm}

\textbf{Step 2:}   \emph{We show in the following that there exists a point $p \in \partial V_0 \cap U_1 \cap U_2 $ such that
\begin{equation}\label{PropBoundaryPoint}
J_{g(u_1, du_1)}^-(p,U_1) \cap  \partial  V_0 \cap J^+_{g(u_1,du_1)}(S, U_1) = \{p\} 
\end{equation}
holds. Such a point $p$ can be thought of as a  point where the boundary is spacelike.}

In the following the causality relations are with respect to the metric $g(u_1, du_1)$. Let $q$ be as above.
If \eqref{PropBoundaryPoint} holds for $q = p$, we are done -- hence, we assume that there is a second point $r \in J^-(q, U_1) \cap \partial V_0 \cap J^+(S, U_1)$. The global hyperbolicity of $V_0$ with Cauchy hypersurface $S$ together with the openness of the timelike relation $\ll$ implies that $\partial V_0 \cap J^+(S, U_1)$ is achronal. Hence, the past directed causal curve $\gamma : [0,1] \to U_1$ with $\gamma (0) = q$ and $\gamma(1) = r$ is a null geodesic. Using the global hyperbolicity of $V_0$ it follows that $\gamma([0,1]) \subseteq \partial V_0 \cap U_1 \cap U_2$, since if there were a $t \in (0,1)$ with $\gamma(t) \in V_0$, then  we would also obtain $\gamma(1) \in V_0$ -- and if $\gamma(t) \in U_1 \setminus V_0$, then by the openness of the timelike relation $\ll$ one could find a past directed timelike curve starting from a point in $V_0$ close to $q$ that lies completely in $J^+(S, U_1)$ and ends at a point in $U_1 \setminus V_0$ close to $\gamma(t)$. Moreover, $\gamma([0,1]) \subseteq U_2$, since $\gamma(0) \in U_2$ and by the smoothness of $u_2$, $\gamma$ is also a past directed null geodesic in the globally hyperbolic $U_2$ -- hence, it cannot leave $U_2$ without first crossing $S$.

We now extend $\gamma$ maximally in $U_1$ to the past. The global hyperbolicity of $U_1$ entails that $\gamma$ has to intersect $S$, thus entering $V_0$ and leaving $\partial V_0$. We now consider $\gamma^{-1}\big(\partial V_0 \cap J^+(S, U_1)\big)$.  The argument from the last paragraph shows that this is a connected interval, and the closedness of $\partial V_0 \cap J^+(S, U_1)$ in $J^+(S, U_1)$ together with $S \subseteq V_0$ implies that $\gamma^{-1}\big(\partial V_0 \cap J^+(S, U_1)\big) =: [0,a]$ for some $a >0$. Note also that it follows from the last paragraph that $\gamma(a) \in U_2$. We claim that $p := \gamma(a) \in J^+(S, U_1)$ satisfies \eqref{PropBoundaryPoint}. 


Assume $p = \gamma(a)$ does not satisfy \eqref{PropBoundaryPoint}. Then there is a point $s \in J^-(p, U_1) \cap \partial V_0 \cap J^+(S, U_1)$. As before, one can connect $p$ and $s$ by a past directed null geodesic that is contained in $\partial V_0 \cap U_1 $. However, by the definition of $a$, this null geodesic cannot be the continuation of $\gamma|_{[0,a]}$. We can thus connect $q$ and $s$ by a broken null geodesic, and thus by a timelike curve in $U_1$ -- contradicting the achronality of $\partial V_0 \cap J^+(S, U_1)$.\vspace*{3mm}

\textbf{Step 3:} \emph{Let $p \in \partial V_0 \cap U_1 \cap U_2$ be as in \eqref{PropBoundaryPoint}. We claim that for every open neighbourhood $W \subseteq U_1$ of $p$ there exists a point $q \in I^+(p, W) $ such that $J^-(q, U_1) \cap (U_1 \setminus V_0) \cap J^+(S, U_1) \subseteq W$ holds.}

To show this, let $p$ be as above and assume the claim was not true. Then there exists a neighbourhood $W \subseteq U_1$ of $p$ such that for all $q \in I^+(p,W)$ there exists a point $\tilde{q} \in J^-(q, U_1) \cap (U_1 \setminus V_0) \cap J^+(S, U_1) \cap (U_1 \setminus W)$. In particular, let us choose a sequence $q_j \in I^+(p,U_1)$ and $\tilde{q}_j \in J^-(q_j, U_1) \cap (U_1 \setminus V_0) \cap J^+(S, U_1) \cap (U_1 \setminus W)$ with $q_j \in I^-(q_0, U_1)$ for all $j \in \N$ and $q_j \to p$. By the global hyperbolicity of $U_1$ we know that $J^-(q_0, U_1) \cap J^+(S, U_1)$ is compact, and thus, so is $J^-(q_0, U_1) \cap (U_1 \setminus V_0) \cap J^+(S, U_1) \cap (U_1 \setminus W)$. Hence, we can assume without loss of generality that $\tilde{q}_j \to \tilde{q}_\infty \in J^-(q_0, U_1) \cap (U_1 \setminus V_0) \cap J^+(S, U_1) \cap (U_1 \setminus W)$. Since the causality relation $\leq$ is closed on globally hyperbolic Lorentzian manifolds, we obtain $\tilde{q}_\infty \leq p$. Moreover, since $\tilde{q}_\infty \in (U_1 \setminus W)$, we clearly have $\tilde{q}_\infty < p$. By \eqref{PropBoundaryPoint} we cannot have $\tilde{q}_\infty \in \partial V_0$, thus we have $\tilde{q}_\infty \in J^-(p,U_1) \cap (U_1 \setminus \overline{V}_0) \cap J^+(S, U_1)$. This, however, contradicts the global hyperbolicity of $V_0$, since, by the openness of the timelike connectedness relation $\ll$, we can find a past directed timelike curve starting at a point contained in $V_0$ close to $p$ and ending at a point in $U_1 \setminus V_0$ without crossing the Cauchy hypersurface $S$.\vspace*{3mm}

\textbf{Step 4:} \emph{We construct a spacelike hypersurface $\Sigma \subseteq \overline{V}_0 \cap U_1 \cap U_2$ that contains at least one point of $\partial V_0 \cap U_1 \cap U_2$.}

Let $p \in \partial V_0 \cap U_1 \cap U_2$ be as in \eqref{PropBoundaryPoint} and consider a convex neighbourhood $W \subseteq I^+(S, U_1)\cap U_2 $ of $p$. By the previous step we can find a point $q \in I^+(p,W)$ such that $J^-(q, U_1) \cap (U_1 \setminus V_0) \cap J^+(S, U_1) \subseteq W$ holds. We denote by $\tau_q : W \to [0, \infty)$ the (past) time separation from $q$ in $W$, i.e., for $r$ in $W$ we have
\begin{equation*}
\tau_q (r) = \sup \{ L(\gamma) \; | \; \gamma \textnormal{ is a past directed timelike curve in $W$ from $q$ to $r$} \}\;,
\end{equation*}
where $L(\gamma)$ denotes the Lorentzian length of $\gamma$. If $r \notin I^-(q, W)$, then we set $\tau_q(r) = 0$. It follows from \cite[Chapter 5, 34. Proposition]{ONeill} that $\tau_q$ restricted to $I^-(q,W)$ is given by
\begin{equation*}
\tau_q(r) = \sqrt{-g|_q\big(\exp^{-1}_q(r), \exp^{-1}_q(r)\big)}\;,
\end{equation*}
hence, $\tau_q$ is smooth on $I^-(q,W)$ and continuous on $W$. Since $J^-(q, U_1) \cap (U_1 \setminus V_0) \cap J^+(S, U_1)$ is compactly contained in $W$, there exists an $r_0 \in J^-(q, U_1) \cap (U_1 \setminus V_0) \cap J^+(S, U_1)$ with
\begin{equation}\label{DefMax}
\tau_q(r_0) = \max \{ \tau_q(r) \; | \; r \in J^-(q, U_1) \cap (U_1 \setminus V_0) \cap J^+(S, U_1)\} =: \tau_0\;.
\end{equation}
Clearly, we have $\tau_0 >0$. We set $\Sigma := \tau_q^{-1}(\tau_0) \subseteq W \subseteq U_1 \cap U_2$. It follows from \cite[Chapter 5, 3. Corollary]{ONeill} that $\Sigma$ is a smooth spacelike hypersurface. Moreover, we have $r_0 \in J^-(q, U_1) \cap \partial V_0 \cap J^+(S, U_1)$, since if we had $r_0 \in J^-(q, U_1) \cap (U_1 \setminus \overline{V}_0) \cap J^+(S, U_1)$, we could extend the unique timelike geodesic from $q$ to $r_0$ slightly such that it still remains in $J^-(q, U_1) \cap (U_1 \setminus V_0) \cap J^+(S, U_1)$, contradicting \eqref{DefMax}. Hence, $\Sigma$ contains at least one point of $\partial V_0 \cap U_1 \cap U_2$. Since we have chosen $W$ to be contained in the future of $S$ in $U_1$, we in particular have $\Sigma \subseteq J^-(q, U_1) \cap J^+(S, U_1)$. Thus, the same argument as before shows $\Sigma \subseteq \overline{V}_0$.\vspace*{3mm}

\textbf{Step 5:} Since $u_1$ and $u_2$ agree on $V_0$, by continuity they (and their derivatives) also agree on $\Sigma \subseteq \overline{V}_0 \cap U_1 \cap U_2$. Consider now a point in $\Sigma \cap \partial V_0 \cap U_1 \cap U_2$ and take a simply connected neighbourhood $W \subseteq U_1 \cap U_2$ thereof such that $\Sigma_W := \Sigma \cap W$ is a closed hypersurface in $W$. By \cite[Chapter 14, 46.\ Corollary]{ONeill}, $\Sigma_W$ is acausal in $W$.

Let $D_1\Sigma_W \subseteq W \subseteq U_1$ denote the domain of dependence of $\Sigma_W$ in $ W \subseteq U_1$ and $D_2\Sigma_W$ the domain of dependence  of $\Sigma_W$ in $W \subseteq U_2$. $u_1|_{D_1\Sigma_W}$ and $u_2|_{D_2\Sigma_W}$ are both globally hyperbolic developments of the same initial data on $\Sigma_W$, and thus by Proposition \ref{PropLocUniqueness} they agree in some small globally hyperbolic neighbourhood $O \subseteq D_1\Sigma_W \cap D_2\Sigma_W$ of $\Sigma_W$. Note that $O$ contains at least one point of $\partial V_0 \cap U_1 \cap U_2$. Moreover, it is easy to see that $V_0 \cup O$ is globally hyperbolic with Cauchy hypersurface $S$: an inextendible causal curve in $O$ has to intersect $\Sigma_W$ and thus, to the past, enter $I^-(\Sigma_W,O) \subseteq V_0$, where the last inclusion follows from $\Sigma_W \subseteq \overline{V}_0 \cap U_1 \cap U_2$ and the global hyperbolicity of $V_0$.  This, however, contradicts the maximality of $V_0$.
\end{proof}

\begin{remark}\label{RemTime}
Let us remark that the proof in particular shows that under the assumptions of Theorem \ref{ThmUnique} the intersection $U_1 \cap U_2$ is the maximal common globally hyperbolic development of $u_1 : U_1 \to \R$ and $u_2 : U_2 \to \R$. If we fix a choice of normal of the initial data hypersurface $S$ and stipulate that it is future directed in $U_1$ as well as in $U_2$, then it follows in particular that a point $x \in U_1 \cap U_2$, which lies to the future  of $S$ in $U_1$, also lies to the future of $S$ in $U_2$. Similarly for the past.
\end{remark}

The following is an immediate consequence of the previous theorem. It shows that global uniqueness can only be violated for quasilinear wave equations in a specific way. 

\begin{corollary} \label{CorDisCon}
Let $u_1 : U_1 \to \R$ and $u_2 : U_2 \to \R$ be two globally hyperbolic developments of \eqref{QuasilinearEq} arising from the same initial data on a connected hypersurface $S \subseteq \R^{d+1}$. If there exists an $x \in U_1 \cap U_2$ with $u_1(x) \neq u_2(x)$, then $U_1 \cap U_2$ is not connected.
\end{corollary}

In particular, we recover that globally defined solutions are unique:
\begin{corollary}
Let $u_1 : \R^{d+1} \to \R$ be a globally defined globally hyperbolic development of \eqref{QuasilinearEq} arising from some initial data on a connected hypersurface $S \subseteq \R^{d+1}$. Let $u_2 : U_2 \to \R$ be another globally hyperbolic development of \eqref{QuasilinearEq} of the same initial data. Then $u_1|_{U_2} = u_2$.
\end{corollary}

In the next two sections we consider two globally hyperbolic developments $u_1 : U_1 \to \R$ and $u_2 : U_2 \to \R$ of the same initial data and discuss criteria that ensure that $U_1 \cap U_2$ is connected. Here, the choice of the initial data hypersurface $S$ plays an important role. This can already be seen from the special case of the linear wave equation in Minkowski space: consider a spacelike \emph{but not achronal}  hypersurface that winds up around the $x^0$-axis in $\R\times (B_2^d(0) \setminus B_1^d(0)) \subseteq \R^{d+1}$. Prescribing generic initial data on this hypersurface, the extent of the future development restricts the extent of the past development. Given two globally hyperbolic developments, their intersection is in general not connected and global uniqueness does not hold. However, it is easy to show (see also Section \ref{SecSupLum}) that for spacelike initial data hypersurfaces which are moreover achronal, this pathology for the linear wave equation in Minkowski space cannot occur. This example shows that any result demonstrating connectedness of $U_1 \cap U_2$ for more general quasilinear equations will require some additional assumptions on the initial surface $S$ analogous to the achronality assumption in Minkowski spacetime. 

\subsection{Uniqueness results for superluminal quasilinear wave equations}
\label{SecSupLum}

In the following we consider quasilinear wave equations \eqref{QuasilinearEq} that enjoy property \eqref{Property}, i.e., that there exists a vector field  $T$ on $\R^{d+1}$ such that $T$ is timelike with respect to  $g^{\mu \nu}(u, d u)$ for all $u, d u$. In particular, superluminal equations enjoy this property, since one can take $T = \partial/\partial x^0$ where $x^\mu$ are inertial frame coordinates.  We will show that for such equations the complication of $U_1 \cap U_2$ being disconnected cannot arise, as long as the initial data is prescribed on a hypersurface $S$ with the property that every maximal integral curve of $T$ intersects $S$ at most once.\footnote{For a superluminal equation with $T=\partial/\partial x^0$, any $S$ which is a Cauchy surface for Minkowski spacetime has this property. Of course $S$ also has to obey the assumptions discussed at the beginning of section \ref{theorems_intro} e.g. $S$ has to be spacelike w.r.t. $g(u,du)$.}

\begin{lemma} \label{LemConnected}
Assume that there exists a vector field  $T$ on $\R^{d+1}$ such that $T$ is timelike with respect to  $g^{\mu \nu}(u, d u)$ for all $u, d u$, where $g$ is as in \eqref{QuasilinearEq}. Let $u_1 : U_1 \to \R$ and $u_2 : U_2 \to \R$ be two globally hyperbolic developments of \eqref{QuasilinearEq} arising from the same initial data on a connected hypersurface $S$ which has the property that every maximal integral curve of $T$ intersects $S$ at most once.  Then $U_1 \cap U_2$ is connected.
\end{lemma}

\begin{proof}


Let $u_1 : U_1 \to \R$, $u_2 : U_2 \to \R$ be two globally hyperbolic developments arising from the same initial data on $S$ and let $x \in U_1 \cap U_2$. Let $\gamma$ be the maximal integral curve of $T$ through $x$. By assumption, $\gamma$ intersects $S$ at most once. Since $\gamma \cap U_1$ and $\gamma \cap U_2$ are timelike curves in $U_1$, $U_2$, respectively, and $U_1$, $U_2$ are globally hyperbolic with Cauchy hypersurface $S$, it follows that $\gamma$ intersects $S$ exactly once and that the portion of $\gamma$ from $x$ to $\gamma \cap S$ is contained in $U_1$ as well as in $U_2$. This shows the connectedness of $U_1 \cap U_2$.

\end{proof}

Let us remark, that one can replace in the above lemma the assumption that  $S$ is a connected hypersurface such that every maximal integral curve of $T$ intersects $S$ at most once, with the assumption that $S$ is a hypersurface that separates $\R^{d+1}$ into two components. We leave the small modification of the proof to the interested reader.

\begin{corollary} \label{CorSupLum}
Assume that there exists a vector field  $T$ on $\R^{d+1}$ such that $T$ is timelike with respect to  $g^{\mu \nu}(u, d u)$ for all $u, d u$, where $g$ is as in \eqref{QuasilinearEq} and that initial data is posed on a connected hypersurface $S$ which has the property that every maximal integral curve of $T$ intersects $S$ at most once. 

Given two globally hyperbolic developments $u_1 : U_1 \to \R$ and $u_2 : U_2 \to \R$, we then have $u_1(x) = u_2(x)$ for all $x \in U_1 \cap U_2$.
\end{corollary}

\begin{proof}
This follows directly from Lemma \ref{LemConnected} and Theorem \ref{ThmUnique}.
\end{proof}

\subsection{Uniqueness results for subluminal quasilinear wave equations}
\label{SecUniquenessSub}

We recall that a quasilinear wave equation of the form \eqref{QuasilinearEq} is called \emph{subluminal} iff the causal cone of $g(u,du)$ is contained inside the causal cone of the Minkowski metric $m = \mathrm{diag}(-1, 1, \ldots, 1)$. 
As shown in Section \ref{SecIVP} of this paper, and in particular see Remark \ref{RemCutOff}, in general global uniqueness does not hold for subluminal quasilinear wave equations -- even if the initial data is posed on the well-behaved hypersurface $\{x^0 =0\}$. However, as we shall show below, developments are unique in regions that are globally hyperbolic with respect to the Minkowski metric. Recall the terminology introduced in Section \ref{SecIVP}: we say that a GHD of a subluminal quasilinear wave equation is a \emph{$m$-GHD} iff it is also globally hyperbolic with respect to the Minkowski metric with Cauchy hypersurface $S$. As usual, $S$ denotes here the initial data hypersurface.

\begin{lemma}
Let $u_1 : U_1 \to \R$ and $u_2 : U_2 \to \R$ be two GHDs of a subluminal quasilinear wave equation \eqref{QuasilinearEq} arising from the same initial data given on a connected hypersurface $S$ that is achronal with respect to the Minkowski metric $m$. Assume, moreover, that $u_1 : U_1 \to \R$ is a $m$-GHD. Then $U_1 \cap U_2$ is connected.
\end{lemma}

\begin{proof}
Let $x \in U_1 \cap U_2$ and assume without loss of generality that $x \in I^+_{g(u_1,du_1)}(S, U_1)$. We claim that this implies $x \in I^+_{g(u_2,du_2)}(S, U_2)$. To see this, assume $x \in I^-_{g(u_2, du_2)}(S, U_2)$. Hence, there exists a future directed timelike curve from $S$ to $x$ in $U_1$ and a future directed timelike curve from $x$ to $S$ in $U_2$. They are both future directed timelike with respect to the Minkowski metric. Concatenating the two curves gives a contradiction to the achronality of $S$ with respect to $m$. This shows $x \in I^+_{g(u_2,du_2)}(S, U_2)$. 
 
Let $\gamma$ be a curve in $U_2$ that starts at $x$ and is timelike, past directed, and past inextendible w.r.t. $g(u_2,du_2)$. It thus intersects $S$.  However, $\gamma$ is also a past directed timelike curve with respect to $m$, and the global hyperbolicity of $U_1$ with respect to $m$ implies that $\gamma$ cannot leave $U_1$ without first intersecting $S$. Thus, the segment of $\gamma$ from $x$ to $S$ is contained in $U_1 \cap U_2$. This shows the connectedness of $U_1 \cap U_2$.
\end{proof}

Together with Theorem \ref{ThmUnique} the above lemma yields

\begin{corollary}
\label{CorUniqueSub}
Let $u_1 : U_1 \to \R$ and $u_2 : U_2 \to \R$ be two GHDs of a subluminal quasilinear wave equation \eqref{QuasilinearEq} arising from the same initial data given on a connected hypersurface $S$ that is achronal with respect to the Minkowski metric $m$. Assume, moreover, that $u_1 : U_1 \to \R$ is a $m$-GHD. Then $u_1 = u_2$ on $U_1 \cap U_2$.
\end{corollary}

\begin{remark}
Let us remark that better bounds on the light cones of $g^{\mu \nu}(u, du)$ translate into an improvement of the uniqueness results. Above, we have only made use of the trivial Minkowski bound on the light cones for subluminal equations. If, for example, for a specific subluminal equation one can improve the a priori bound on the light cones of $g^{\mu \nu}(u, du)$ for certain initial data, then one can also improve the uniqueness result for these initial data.
\end{remark}

\subsection{Local existence for general quasilinear wave equations} \label{SecEx}

This section provides the other half of the local well-posedness statement for quasilinear wave equations with data on \emph{general} hypersurfaces: the local existence result. 

\begin{theorem}[Local existence] \label{ThmEx}
Given initial data for a quasilinear wave equation \eqref{QuasilinearEq}, there exists a globally hyperbolic development.
\end{theorem}

Moreover, this result is needed for the \emph{existence} results of a unique maximal GHD for superluminal quasilinear wave equations and of a maximal unique GHD for subluminal quasilinear wave equations. 

Note that Theorem \ref{ThmEx} with data on the hypersurface $\{t=0\}$ is a standard literature result. We prove Theorem \ref{ThmEx} by using the standard literature result to construct solutions in local coordinate neighbourhoods around points of the general initial data hypersurface and then patching them together. Note that this has to be carried out carefully to ensure that different local solutions agree on the intersection of their domains. Here we make use of Theorem \ref{ThmUnique} to guarantee uniqueness if the intersection of their domains is connected.


\begin{proof}
Given the initial data $f_0, \alpha_0$ on the hypersurface $S$ (as discussed in Section \ref{theorems_intro}) we choose a timelike normal $N$ along $S$ and extend it smoothly off $S$ to yield a vector field which we also denote with $N$. There exists an open neighbourhood $D$ of $\{0\} \times S$ in $\R \times S$ and an open neighbourhood $T \subseteq \R^{d+1}$ of $S$ such that the flow $\Phi$ of $N$ is a diffeomorphism from $D$ onto $T$. For $p \in S$ let $W_p \subseteq T \subseteq \R^{d+1}$ be an open neighbourhood of $p$ on which there exists slice coordinates in which the Lorentzian metric $g(f_0, \alpha_0)$ determined by the initial data is $C^0$-close to the Minkowski metric. Let $S_p \subseteq S$ be a neighbourhood of $p$ in $S$ with closure that is compactly contained in $W_p$. The standard energy methods in the literature (see for example \cite{Sogge}) yield that there exists a globally hyperbolic development $u_p : DS_p \to \R$ for \eqref{QuasilinearEq} of the initial data on $S_p$, where $DS_p \subseteq W_p$. Moreover, by choosing $DS_p$ smaller if necessary, we can assume that $N$ is timelike on $DS_p$. We now claim that for all $p, q \in S$ we have $u_p = u_q$ on $DS_p \cap DS_q$.

To show this, assume that $DS_p \cap DS_q \neq \emptyset$ and let $A$ be a connected component of $DS_p \cap DS_q $. Consider an $x \in A$. The integral curve of $N$ through $x$ is a timelike curve in $DS_p$ as well as in $DS_q$, and thus it has to intersect $S_p \cap S_q$ and, moreover, its segment from $x$ to $S_p \cap S_q$ is contained in $DS_p \cap DS_q$. This shows that $A \cap (S_p \cap S_q)$ is non-empty. It will follow a posteriori that $A \cap (S_p \cap S_q)$ is connected, but for the time being let $S_A$ be a connected component of $A \cap (S_p \cap S_q)$. We denote with $D_pS_A, D_qS_A$ the domain of dependence of $S_A$ in $A$ with respect to the Lorentzian metric arising from $u_p$ and $u_q$, respectively. Since by the above argument involving the timelike integral curves of $N$, the intersection $D_p S_A \cap D_q S_A$ is connected, Theorem \ref{ThmUnique} implies that we have $u_p = u_q$ on $D_p S_A \cap D_q S_A$. 

Assume now that $D_p S_A \cap D_q S_A \subsetneq A$. Since $A$ is connected, there exists an $r \in \partial (D_p S_A \cap D_q S_A) \cap A$. Without loss of generality we assume that $r$ lies to the future of $S_A$. Let $\gamma$ be any past directed and past inextendible causal curve in $A$ with respect to the metric arising from $u_p$ that starts at $r$. The global hyperbolicity of $D_pS_A \cap D_q S_A$ (see Remark \ref{RemTime}) implies
\begin{equation}
\label{AuxEq}
J^+_{g(u_p, du_p)}(S_A, A) \cap \mathrm{Im}(\gamma) \subseteq \overline{D_pS_A \cap D_qS_A} \;,
\end{equation}
and hence the part of $\gamma$ to the causal future of $S_A$ is also a past directed causal curve in $A$ with respect to $u_q$. The global hyperbolicity of $DS_p$ and $D S_q$ shows that $\gamma$ has to intersect $S_p \cap S_q$, and by \eqref{AuxEq}, $\gamma$ in fact intersects $S_A$. This, however, gives the contradiction $r \in D_pS_A \cap D_qS_A$ by definition of the domain of dependence.  We thus conclude that $D_p S_A \cap D_q S_A = A$. Moreover, it now follows that $u_p = u_q$ holds on $DS_p \cap DS_q$. Hence, we can finish the proof by constructing a GHD $u : U \to \R$ of the given initial data on $S$ by setting $U = \bigcup_{p \in S} DS_p$ and $u(x) = u_p(x)$ for $x \in DS_p$.
\end{proof}

\subsection{The existence of a unique maximal GHD for superluminal quasilinear wave equations} \label{SecUMGHD}

\begin{theorem} \label{ThmUMGHD}
Assume that there exists a vector field  $T$ on $\R^{d+1}$ such that $T$ is timelike with respect to  $g^{\mu \nu}(u, d u)$ for all $u, d u$, where $g$ is as in \eqref{QuasilinearEq} and that initial data is posed on a connected hypersurface $S$ which has the property that every maximal integral curve of $T$ intersects $S$ at most once.

Given such initial data, there then exists a \emph{unique maximal globally hyperbolic development} $u_{\max} : U_{\max} \to \R$, that is, a globally hyperbolic development  $u_{\max} : U_{\max} \to \R$ with the property that for any other globally hyperbolic development $u : U \to \R$ of the same initial data we have $U \subseteq U_{\max}$ and $u_{\max}|_U = u$.
\end{theorem}

\begin{proof}
We consider the set $\{(u_\alpha, U_\alpha) \; | \; \alpha \in A\}$ of all globally hyperbolic developments $u_\alpha : U_\alpha \to \R$ arising from the given initial data on $S$ as above. Note that this is a set and, moreover, it is non-empty by Theorem \ref{ThmEx}. We now define $U_{\max} := \bigcup_{\alpha \in A} U_\alpha$ and $u_{\max} : U \to \R$ by $u_{\max}(x) = u_\alpha(x)$ for $\alpha \in A$ with $x \in U_\alpha$. Note that the latter is well-defined by Corollary \ref{CorSupLum}. In order to see that $u_{\max} : U_{\max} \to \R$ is a \emph{globally hyperbolic} development of \eqref{QuasilinearEq} arising from the given initial data, consider an inextendible timelike curve $\gamma : (a,b) \to U_{\max}$, where $-\infty \leq a < b \leq \infty$, and let $a < t_0 < b$. We have $\gamma(t_0) \in U_{\alpha_0}$ for some $\alpha_0 \in A$. Let $(a_0, a_1) \subseteq I$ be the maximal interval containing $t_0$ such that $\gamma|_{(a_0,a_1)}$ maps into $U_{\alpha_0}$. Since $\gamma|_{(a_0,a_1)}$ is an inextendible timelike curve in $U_{\alpha_0}$, there exists a $\tau_0 \in (a_0, a_1)$ with $\gamma(\tau_0) \in S$. Thus, it remains to show that $\gamma$ does not intersect $S$ more than once. Without loss of generality we assume that $\gamma_{(a_0, a_1)}$ is future directed in $U_{\alpha_0}$. We consider
\begin{equation*}
J = \{ t \in (\tau_0, b) \; | \; \exists \alpha \in A \textnormal{ with } \gamma \big( [ \tau_0, t]\big) \subseteq U_\alpha\} \;.
\end{equation*}
We already know that $J$ is non-empty. Moreover, $J$ is clearly open, since each $U_\alpha$ is open. Let $t_n \in J$ be a sequence with $t_n \to t_\infty \in (\tau_0, b)$ as $n \to \infty$, and let $\alpha_\infty \in A$ be such that $\gamma(t_\infty) \in U_{\alpha_\infty}$. By the openness of $U_{\alpha_\infty}$ there is $n_0 \in \N$ with $\gamma(t_{n_0}) \in U_{\alpha_\infty}$. Since $t_{n_0} \in J$, there exists $\alpha_{n_0} \in A$ with $\gamma \big( [ \tau_0, t_{n_0}]\big) \subseteq U_{\alpha_{n_0}}$. It now follows from Remark \ref{RemTime} that $\gamma(t_{n_0})$ must also lie to the future of $S$ in $U_{\alpha_\infty}$. Hence, $S$ being a Cauchy hypersurface of $U_{\alpha_\infty}$ implies that $t_\infty \in J$. It thus follows that $J = (\tau_0, b)$. We conclude that $\gamma$ cannot intersect $S$ again to the future of $\tau_0$. The analogous argument shows that it can neither intersect $S$ again to the past of $\tau_0$. We thus conclude that $U_{\max}$ is globally hyperbolic with Cauchy hypersurface $S$.

Finally, it is clear that any other globally hyperbolic development of the same initial data is contained in $U_{\max}$.
\end{proof}

\begin{remark} \label{RemMGHD}
We note that the above construction of a unique maximal globally hyperbolic development is always possible provided the property of global uniqueness holds.
\end{remark}

\subsection{The existence of a maximal unique GHD for subluminal quasilinear wave equations} \label{SecMUGHD}

As mentioned before, for subluminal quasilinear wave equations there does not generally exist a unique maximal globally hyperbolic development. In this section we show existence of a globally hyperbolic development on the domain of which the solution is uniquely defined and which is maximal among all GHDs that have this property. But first we establish some terminology: We consider a subluminal quasilinear wave equation of the form \eqref{QuasilinearEq} and consider initial data prescribed on a connected hypersurface $S$ that is \emph{acausal} with respect to the Minkowski metric $m$, i.e., there does not exist a pair of points on $S$ that can be connected by a causal curve within the Minkowski spacetime. We call a GHD $u_1 : U_1 \to \R$ a \emph{unique globally hyperbolic development} (UGHD) iff for all other GHDs $u_2 : U_2 \to \R$ we have $u_1 = u_2$ on $U_1 \cap U_2$. We note that any $m$-GHD is a UGHD by Corollary \ref{CorUniqueSub}.

\begin{theorem} \label{ThmMUGHD}
Consider a subluminal quasilinear wave equation of the form \eqref{QuasilinearEq}.
Given initial data on a connected hypersurface $S$ that is acausal with respect to the Minkowski metric there exists a UGHD $u : U \to \R$ with the property that the domain of any other UGHD is contained in $U$. The UGHD $u : U \to \R$ is called the \emph{maximal unique globally hyperbolic development} (MUGHD).
\end{theorem}

\begin{proof}
We consider the set $\{u_\alpha : U_\alpha \to \R \; | \; \alpha \in A\}$ of all UGHDs of the given initial data. Note that this set is non-empty: by Theorem \ref{ThmEx} there exists a GHD $u_1 : U_1 \to \R$ and we can now consider the domain of dependence of $S$ in $U_1$ with respect to the Minkowski metric. By \cite[Chapter 14, 38.\ Theorem and 43.\ Lemma]{ONeill} this gives rise to a $m$-GHD. By Corollary \ref{CorUniqueSub} this is a UGHD.

We now set $U := \bigcup_{\alpha \in A} U_\alpha$ and $u(x) := u_\alpha(x)$ for $x \in U_\alpha$. The latter is well-defined since each $u_\alpha : U_\alpha \to \R$ is a UGHD. The same argument as in the proof of Theorem \ref{ThmUMGHD} shows that $u : U \to \R$ is a \emph{globally hyperbolic} development. To show that it is a UGHD, let $u_2 : U_2 \to \R$ be a GHD  and consider $x \in U \cap U_2$. There exists an $\alpha \in A$ with $x \in U_\alpha$, and since $u_\alpha : U_\alpha \to \R$ is a UGHD it follows that $u_2(x) = u_\alpha(x) = u(x)$. Finally, it is clear by construction that the domain of any other UGHD is contained in $U$.
\end{proof}

\begin{theorem}
Consider a subluminal quasilinear wave equation of the form \eqref{QuasilinearEq}.
Given initial data on a connected hypersurface $S$ that is acausal with respect to the Minkowski metric there exists a unique $m$-MGHD, i.e., a $m$-GHD $u_{\max}: U_{\max} \to \R$  with the property that for any other $m$-GHD $u: U \to \R$ of the same initial data we have $U \subseteq U_{\max}$ and $u_{\max}|_U = u$. 
\end{theorem}

\begin{proof}
One considers the set of all $m$-GHDs of the given initial data. The beginning of the proof of Theorem \ref{ThmMUGHD} shows that this set is non-empty. Using Corollary \ref{CorUniqueSub}, which provides a global uniqueness statement for $m$-GHDs, Remark \ref{RemMGHD} shows that one can now proceed as in the proof of Theorem \ref{ThmUMGHD} to construct the maximal element in the above set of all $m$-GHDs.
\end{proof}

We summarise that given a GHD for a \emph{superluminal} equation, one knows that it is contained in the unique maximal GHD. For \emph{subluminal} equations, there are in general GHDs which are not contained in the maximal UGHD. However, given a $m$-GHD, it \emph{is} contained in the maximal UGHD. In particular the $m$-MGHD is contained in the MUGHD, but in general the latter is strictly bigger.

Let us also remark that we expect that the analogue of Theorem  \ref{ThmMUGHD} does not hold for more general quasilinear wave equations, i.e., ones which are neither subluminal nor superluminal. Indeed, even more strongly, we formulate the following 
\begin{conjecture} \label{ConUni}
There are quasilinear wave equations of the form \eqref{QuasilinearEq} for which there exists initial data such that there does not exist \emph{any} UGHD.
\end{conjecture}
This conjecture is based on the following scenario which we think might happen: there exists a quasilinear wave equation of the form \eqref{QuasilinearEq} and initial data such that there exists an infinite family of GHDs the domains of which bend round back towards the initial data hypersurface $S$ and approach it arbitrarily closely, as shown in Figure \ref{FigCon}. This would imply that there is no neighbourhood of $S$ on which the solution is uniquely defined. In particular, this would establish the sharpness of the local uniqueness statement of Proposition \ref{PropLocUniqueness}.

\begin{figure}[h]
  \centering
  \includegraphics[width=8cm]{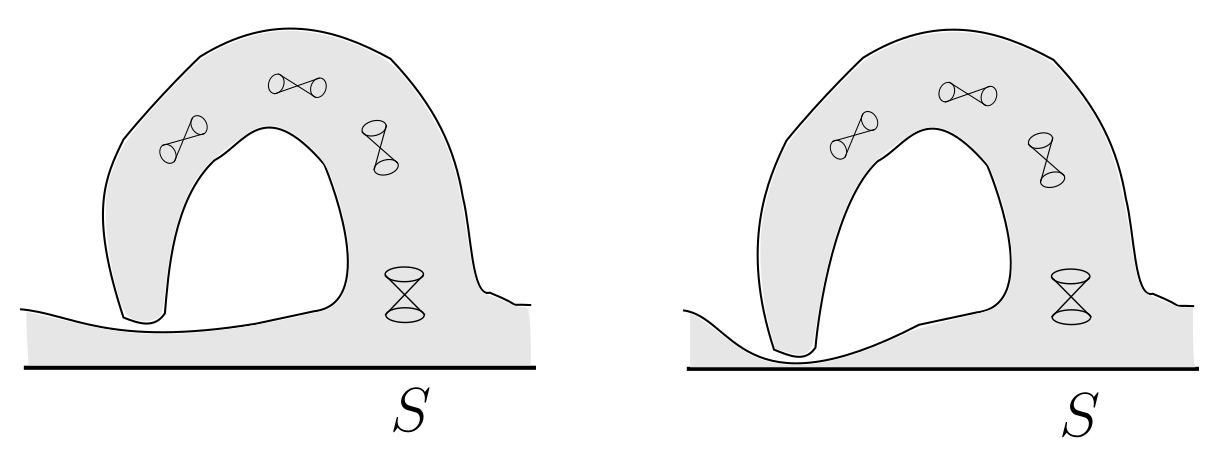}
      \caption{A possible mechanism for a resolution of Conjecture \ref{ConUni}. The Figure shows the light cones of $g(u,du)$.} \label{FigCon}
\end{figure}

\subsection{A uniqueness criterion for general quasilinear wave equations at the level of MGHDs}
\label{SecFin}

In this section we consider a general quasilinear wave equation of the form \eqref{QuasilinearEq}. Recall that a GHD $u_1 : U_1 \to \R$ of given initial data posed on a hypersurface $S$ is called a \emph{maximal globally hyperbolic development} (MGHD) iff there does not exist a GHD $u_2 : U_2 \to \R$  of the same initial data with $U_1 \subsetneq U_2$. Note that by Theorem \ref{ThmUnique} any such GHD $u_2 : U_2 \to \R$ would agree with $u_1$ on $U_1$, and thus it would correspond to an extension of $u_1 : U_1 \to \R$. In other words, a MGHD is a GHD that cannot be extended as a GHD. 

The example from Section \ref{SecIVP} shows that in general there can exist infinitely many MGHDs for given initial data. Consider now two such MGHDs $u_1 : U_1 \to \R$ and $u_2 : U_2 \to \R$ arising in the example of Section \ref{SecIVP}. Then $U_1 \cap U_2$ is disconnected. Let $A$ denote the connected component containing $S$. Consider a point $x \in U_1 \cap U_2$ which does not lie in $A$. 
\emph{The phenomenon of non-uniqueness, i.e., that $u_1(x)$ does not equal $u_2(x)$, arises, because the `path of evolution' the second solution takes from $A$ to reach $x$ is blocked because the first solution is already defined in that very region.} In the example of Section \ref{SecIVP}, this behaviour arises because $U_1$ (say) lies ``on both sides of its boundary". The following theorem makes this precise and shows that this is the only mechanism at the level of MGHDs that leads to non-uniqueness for general quasilinear wave equations. It states that given an MGHD with the property that its domain of definition always lies to \emph{just one side} of its boundary, i.e., the domain of definition cannot block evolution elsewhere, then it is the unique MGHD.

\begin{theorem} \label{ThmNotBlocking}
Let $u_1 : U_1 \to \R$ be a MGHD of given initial data for a quasilinear wave equation of the form \eqref{QuasilinearEq} and assume that 
\begin{equation}
\label{PropertyBoundary}
\parbox{0.75\textwidth}{for every $p \in (\partial U_1 \setminus \partial S)$ there exists a neighbourhood $V$ of $p$ together with a chart $\psi : V \to (-\varepsilon, \varepsilon)^{d+1}$, $\varepsilon >0$, and a continuous function $f : (-\varepsilon, \varepsilon)^d \to (-\varepsilon, \varepsilon)$ such that $\psi^{-1}(\mathrm{graph} f) = \partial U_1 \cap V$, all points below $\mathrm{graph} f$ in $(-\varepsilon, \varepsilon)^{d+1}$ are mapped into $U_1$ and all points above $\mathrm{graph} f$ in $(-\varepsilon, \varepsilon)^{d+1}$ are mapped into $\R^{d+1} \setminus U_1$.}
\end{equation} 
Then $u_1 : U_1 \to \R$ is the unique MGHD, i.e., any other GHD $u : U \to \R$ satisfies  $U \subseteq U_1$ and thus also $u_1|_U = u$.
\end{theorem} 

Note that in order to apply this theorem to a concrete example one has to first construct a/the whole MGHD and is only then able to infer {\it a posteriori} that the evolution was indeed unique.

\begin{proof}
Let $u_1 : U_1 \to \R$ be a MGHD of given initial data such that \eqref{PropertyBoundary} is satisfied.
Let $u_2 : U_2 \to \R$ be a second GHD of the same initial data and, to obtain a contradiction, we assume that 
$U_2 \nsubseteq U_1$. Let us denote the connected component of $U_1 \cap U_2$ that contains the initial data hypersurface $S$ with $A$. A point in the boundary of $\partial A$ cannot be contained in $U_1$ as well as in $U_2$ by definition of $A$.
Since we have 
$U_2 \nsubseteq U_1$
it follows that $\partial A \cap U_2$ is non-empty and contained in the complement of $U_1$. Thus, we obtain
\begin{equation} \label{BoundaryNE}
\emptyset \neq \partial A \cap U_2 \subseteq \partial U_1 \;.
\end{equation}
Hence, we have exhibited a part of the boundary of the MGHD $u_1 : U_1 \to \R$ to which the solution extends smoothly (from $A$). The idea is now to use property \eqref{PropertyBoundary} to show that one can actually extend $u_1$ across this boundary to obtain a bigger GHD -- thus violating the maximality of $u_1$.\footnote{In general, i.e., if property \eqref{PropertyBoundary} is not satisfied, this might not be possible since there is no free space on the other side of the boundary to construct an extension.} The construction is similar to the on in the proof of Theorem \ref{ThmUnique}.

A slight variation of Remark \ref{RemTime} shows that the set $A$ is the MCGHD of $u_1 : U_1 \to \R$ and $u_2 : U_2 \to \R$. In particular, $A$ is globally hyperbolic with Cauchy surface $S$. By \eqref{BoundaryNE}, let $q \in \partial A \cap U_2$ and assume without loss of generality that $q \in J^+_{g(u_2,du_2)}(S,U_2)$. We are going to show that there exists a point $p \in \partial A \cap U_2 \cap  J^+_{g(u_2,du_2)}(S,U_2)$ with
\begin{equation}
\label{SpacelikeP}
J^-_{g(u_2,du_2)}(p, U_2) \cap \partial A \cap  J^+_{g(u_2,du_2)}(S,U_2) = \{p\} \;.
\end{equation}
The proof of this is analogous to Step 2 in the proof of Theorem \ref{ThmUnique} and is only sketched in the following. Assume \eqref{SpacelikeP} does not hold for $p = q$. Then there exists another point $r \in J^-_{g(u_2,du_2)}(p, U_2) \cap \partial A \cap  J^+_{g(u_2,du_2)}(S,U_2)$. The global hyperbolicity of $A$ implies the achronality of $\partial A \cap J^+_{g(u_2,du_2)}(S,U_2)$. Hence, the past directed causal curve connecting $q$ with $r$ is a null geodesic which lies in $\partial A \cap   J^+_{g(u_2,du_2)}(S,U_2)$. We now extend this null geodesic maximally to the past and consider the point $p$ where it leaves $\partial A \cap  J^+_{g(u_2,du_2)}(S,U_2)$. This point $p$ satisfies \eqref{SpacelikeP}. 

Step 3 of the proof of Theorem \ref{ThmUnique} applies literally unchanged if $V_0$ is replaced by $A$. Following Step 4 of the proof of Theorem \ref{ThmUnique} we now construct a spacelike  (with respect to $g(u_2, du_2)$) hypersurface $\Sigma \subseteq \overline{A} \cap U_2$ that contains at least one point  $q \in \partial A \cap U_2 \subseteq \partial U_1$.

By \eqref{PropertyBoundary} we can now find a neighbourhood $V$ of $q$ together with a chart $\psi : V \to (-\varepsilon, \varepsilon)^{d+1}$ and a continuous function $f : (-\varepsilon, \varepsilon)^d \to (-\varepsilon, \varepsilon)$ such that in this chart $\partial U_1 \cap V$ is given by the graph of $f$, $U_1 \cap V$ lies below the graph of $f$, and $V \setminus U_1$ lies above the graph of $f$. We can, after making $V$ smaller if necessary, assume that $V\subseteq U_2$ and that the spacelike hypersurface $\Sigma_V := \Sigma \cap V$ is a closed hypersurface in $V$. It follows from \cite[Chapter 14, 46.\ Corollary]{ONeill} that $\Sigma_V$ is acausal in $V \subseteq U_2$.
We consider now the  domain of dependence $D\Sigma_V$ of $\Sigma_V$ in $V \subseteq U_2$. Clearly, $D\Sigma_V$ contains points that lie above the graph of $f$ in the chart $\psi$. We can now define $u_3 : U_3 \to \R$, $U_3 := U_1 \cup D\Sigma_V$, $u_3(x) := u_1(x)$ for $x \in U_1$ and $u_3(x) := u_2(x)$ for $x \in D\Sigma_V$. This is well defined since the region below the graph of $f$ in the chart $\psi$ lies in $A$, where $u_1$ and $u_2$ agree. It is easy to see that $u_3 : U_3 \to \R$ is a GHD the domain of which contains that of the MGHD $u_1 : U_1 \to \R$. This is a contradiction.
\end{proof}

We conclude with presenting a simple criterion that ensures that condition \eqref{PropertyBoundary} is satisfied. It is tailored to small data results.

\begin{lemma}
Let $u_1 : U_1 \to \R$ be a GHD of given initial data posed on an open and connected subset $S$ of $\{x^0=0\}$ for a  quasilinear wave equation of the form \eqref{QuasilinearEq}. Furthermore, assume that 
\begin{equation}
\label{PropertyTimelike}
\parbox{0.75\textwidth}{
there exists a $\delta >0$ such that $\partial_0 + \sum_{i=1}^d \delta_i \partial_i$ is timelike with respect to $g(u_1, du_1)$ for all $\delta_i \in \R$ with $\sum_{i=1}^d |\delta_i| < \delta$.}
\end{equation} 
Then the condition \eqref{PropertyBoundary} is satisfied.
\end{lemma}

As an application of the lemma and of Theorem \ref{ThmNotBlocking} let us mention the work \cite{Chr} of Christodoulou in which he studies the formation of shocks for relativistic perfect fluids. In the irrotational case the equations of motion give rise to a subluminal wave equation. For sufficienly small initial data he explicitly constructs a MGHD and Theorem 13.1, conclusion iii) in \cite{Chr} shows that the assumptions of the above lemma are met. 

Before we give the proof, let us also emphasise that condition \eqref{PropertyTimelike} only ensures that $u_1 : U_1 \to \R$ is a UGHD \emph{if it is a MGHD to start with}.

\begin{proof} Let us introduce the notation $x = (x^0, x^1, \ldots x^d) = (x^0, \underline{x})$ with $\underline{x} \in \R^d$ and
let $p = (t_0, \underline{x}_0) \in \partial U_1 \setminus \partial S$. It thus follows that $t_0 \neq 0$. Without loss of generality let us assume that $t_0 >0$ and that $\partial_0$ is future directed.  We first show that $[0,t_0) \times \{\underline{x}_0\} \subseteq U_1$. 

Assume it was not the case and there existed a $0 \leq t_1 < t_0$ with $(t_1, \underline{x}_0) \notin U_1$. We can then find a point $q \in U_1$ sufficiently close to $p$ that can be connected to $(t_1, \underline{x}_0)$ by a straight line with slope at most $\delta$ with respect to the $x^0$-axis, i.e., by a straight line with tangent vector proportional to $\partial_0 + \sum_{i=1}^d \delta_i \partial_i$ for some $\delta_i \in \R$, $\sum_{i=1}^d |\delta_i| < \delta$. This however gives rise to a past directed and past inextendible timelike curve in $U_1$ starting at\footnote{It is clear that $q$ must lie in the future of $S$ in $U_1$, since if it were lying in the past, then the future directed timelike curve with velocity $\partial_0$ starting at $q$ would give rise to a future inextendible curve that does not intersect $S$.}  $q \in I^+_{g(u_1, du_1)}(S,U_1)$, which does not intersect $S$. This contradicts $S$ being a Cauchy hypersurface.

In particular, it follows that $(0,  \underline{x}_0) \in S$. Let now $W$ be a small neighbourhood of $\underline{x}_0$ in $S$ and define $f^+ : W \to \R$ by 
\begin{equation*}
f^+(\underline{x}) := \sup \{t >0 \; | \; (t', \underline{x}) \in U_1 \quad \forall\; 0 \leq t' < t\} \;.
\end{equation*}
Clearly we have $f^+(\underline{x}_0) = t_0$. Note that $f^+(\underline{x})$ is indeed finite for all $\underline{x} \in W$: if it were infinite for some $\underline{x} \in W$, then we could choose $t>0$ large enough such that we could connect $(t, \underline{x}) \in I^+_{g(u_1, du_1)}(S,U_1)$ with $p$ by a straight line with slope at most $\delta$ with respect to the $x^0$-axis -- obtaining a contradiction as before. Indeed, the same kind of argument shows that for a sequence of points $\underline{x}_n \in W$, with $\underline{x}_n \to \underline{x} \in W$ for $n \to \infty$ we must have $f^+(\underline{x}_n) \to f^+(\underline{x})$ for $n \to \infty$, since if there were there were infinitely many $n$ such that $|f^+(\underline{x}_n) - f^+(\underline{x})| \geq \varepsilon_0 >0$ for some $\varepsilon_0 >0$, then we could again construct past inextendible timelike curves in $U_1$ starting in the future of $S$ that do not cross $S$. Hence, $f^+$ is continuous. This kind of argument also immediately shows that $(x^0, f^+(\underline{x})) \notin U_1$ for $x^0 \geq f^+(\underline{x})$. This completes the proof.
\end{proof}

\section*{Acknowledgements}

We are grateful to M. Dafermos for useful discussions. FCE and HSR are funded by STFC. 

\appendix

\section*{Appendix: superluminal equations in two dimensions}

\label{sec:2dsuper}

In this Appendix we will consider causal properties of superluminal equations in $1+1$ dimensions. The low dimensionality imposes strong restrictions on the causal structure of solutions. We will review some results on causality in $1+1$ dimensions and explain why it is not possible to violate causality in a smooth way in a finite region of spacetime. 

Assume we have a hyperbolic solution $u$ defined on some open subset $M$ of $\mathbb{R}^2$. 
If $M$ is simply connected then $M$ is homeomorphic to $\mathbb{R}^2$, which implies that $(M,g)$ is stably causal \cite[Theorem 3.43]{beem}. Hence a violation of stable causality requires that $M$ is not simply connected i.e. $M$ must have holes or punctures. We assume that $(M,u)$ is {\it inextendible}, i.e., it is not possible to extend $u$ as a hyperbolic solution onto a connected open set strictly larger than $M$. Hence the non-trivial topology of $M$ must be associated with $u$ developing some pathological feature when we attempt to extend to points of $\partial M$, for example, $u$ or its derivative might blow up, or $g$ might fail to be Lorentzian at such points. 

This looks bad for the possibility of smoothly violating causality (i.e. ``forming a time machine"). But maybe the above pathological features are {\it consequences} of the time machine, i.e., they lie to the future of the causality violating region. This is not the case: we will explain why some such pathology must occur {\it before} causality is violated. One cannot form a time machine smoothly in a two-dimensional superluminal theory.  

Pick inertial coordinates $(t,x)$ for Minkowski spacetime so that
\be
 m =-dt^2 + dx^2
\ee
Now $dx$ is spacelike w.r.t. $m^{\mu\nu}$, which implies that it is also spacelike w.r.t. $g^{\mu\nu}$ (because, for a superluminal equation, the null cone of $g^{\mu\nu}$ lies on, or inside, that of $m^{\mu\nu}$). Hence $x$ is a {\it global space function} for $g$, i.e., a function with everywhere spacelike (non-zero) gradient. The transformation of the previous subsection relates it to a global time function of the corresponding subluminal equation.

Consider a null geodesic of $g$. Then $x$ must be monotonic along the geodesic. To see this, let $V$ be tangent to the geodesic. Then $V^x = V^\mu (dx)_\mu$ and this cannot vanish because $V$ is null w.r.t. $g$ and $dx$ is spacelike w.r.t. $g$.\footnote{In 2d let $P$ and $Q$ be non-zero vectors such that $g_{\mu\nu} P^\mu Q^\nu = 0$. If $P$ is timelike (spacelike) w.r.t. $g$ then $Q$ must be spacelike (timelike) w.r.t. $g$. If $P$ is null w.r.t. $g$ then $Q$ must also be null, and parallel to $P$.} Non-vanishing of $V^x$ implies that $x$ is monotonic along the geodesic. 
It follows that {\it a null geodesic of $g$ cannot be closed and cannot intersect itself.} (Again this is easy to understand using the transformation of the previous section.)

It is also easy to see that there cannot be a smooth closed future-directed causal curve (w.r.t. $g$) which is {\it simple}, i.e., does not intersect itself. This is because there will be a point on any such curve at which the tangent vector is timelike and past directed w.r.t. the Minkowski metric $m_{\mu\nu}$ and hence also timelike and past directed w.r.t. $g_{\mu\nu}$, contradicting the fact that the curve is future-directed. Hence {\it a closed future-directed causal curve must be non-smooth or non-simple.} 

Now let $S$ be a partial Cauchy surface, i.e., a surface (actually a line) which, viewed as a subset of $(M,g)$, is closed, achronal and edgeless \cite{wald}. The future domain of dependence of $S$ is $D^+(S)$ and the future Cauchy horizon is $H^+(S) = \overline{D^+(S)} - I^-(D^+(S))$. If causality is violated to the future of $S$ then this must occur outside $D^+(S)$, so $H^+(S)$ is non-empty. A standard result states that $H^+(S)$ is achronal and closed, and that every $p \in H^+(S)$ lies on a null geodesic contained in $H^+(S)$ which is past inextendible without a past endpoint in $M$ \cite{wald}.

Consider following a generator of $H^+(S)$ to the past. Since $x$ is monotonic it must either diverge or approach a finite limit along this generator. If $x$ diverges then the generator originates from infinity in $\mathbb{R}^2$. Consider the case that $x$ approaches a finite limit in the past. From the fact that the generator is null w.r.t. $g$ and hence non-timelike w.r.t. $m$ we have $|dt/dx| \le 1$, which implies (via integration) that $t$ also approaches a finite limit. Hence the generator has an endpoint $p$ in $\mathbb{R}^2$. But it cannot have an endpoint in $M$ so $p \notin M$. Since we are assuming that $(M,u)$ is inextendible, $p$ must correspond either to a singularity of the spacetime $(M,g)$, or to a ``point at infinity" in $(M,g)$. In the latter case, $g$ would have to blow up at $p$, which is singular behaviour from the point of view of the Minkowski spacetime.  

This proves that generators of $H^+(S)$ must emanate either from infinity in Minkowski spacetime or from a point of $\mathbb{R}^2$ that is singular w.r.t. $(M,g)$ or ``at infinity" w.r.t. $(M,g)$. None of these possibilities corresponds to what is usually regarded as the condition for creation of a time machine in a bounded region of space, namely a ``compactly generated" Cauchy horizon \cite{hawking} (one whose generators remain in a compact region of $(M,g)$ when extended to the past). If the generator does not emanate from infinity in $\mathbb{R}^2$ then it remains in a compact region of $\mathbb{R}^2$ but not a compact region of $M$: in $M$ it ``emerges from a singularity" or ``from infinity". 

To violate causality in a smooth way, the generators of $H^+(S)$ would have to emanate from infinity. This can happen even for the linear wave equation if $S$ extends to left and/or right past null infinity in 2d Minkowski spacetime. In this case, $H^+(S)$ exists because information can enter the spacetime from past null infinity without crossing $S$. This is rather uninteresting (unrelated to any violation of causality) so consider instead the case of $S$ extending to (left and right) spatial infinity in 2d Minkowski spacetime. For such $S$ there is no Cauchy horizon for the linear wave equation so now consider such $S$ for a nonlinear equation of the form (\ref{eqndims}). Assume that the initial data $(u,du)$ is compactly supported on $S$. Under time evolution, the $u$ field can propagate out to future null infinity. In $2d$, even for the linear wave equation solutions do not decay at null infinity, so $u$ does not necessarily decay near future null infinity. This implies that $g$ may not approach $m$ near future null infinity. So perhaps causality violation could originate at infinity with a Cauchy horizon forming at left and/or right future null infinity and propagate into the interior of the spacetime along null geodesics of $g$ which are spacelike w.r.t. $m$. It would be interesting to find an example for which this behaviour occurs.

\end{document}